%% file: twoPhase_porousMedia.tex
\documentclass[a4paper, 12pt, reqno]{amsart}
\usepackage{amsaddr} 

\usepackage[utf8x]{inputenc}
\usepackage[T1]{fontenc}
\usepackage[english]{babel}
\usepackage{amsmath, amsfonts, amsthm, amssymb, amscd}
\usepackage[final]{graphicx}
\usepackage[below]{placeins}
\usepackage{subcaption}
\usepackage{multirow}
\usepackage{mathtools}

\usepackage{etex}

\usepackage[hidelinks]{hyperref}
\hypersetup{
  pdfauthor = {Stefan Metzger},
  pdftitle = {Homogenization of two-phase flow}
}

\usepackage[final]{showkeys} 
\renewcommand\showkeyslabelformat[1]{}
\usepackage{esint}
\usepackage{caption}
\usepackage{dsfont}

\usepackage{tikz}
\usetikzlibrary{patterns}
\usetikzlibrary{shapes.geometric}

\usepackage[]{todonotes}

\usepackage{diagbox}

\newtheorem{theorem}{Theorem}[section]
\newtheorem{lemma}[theorem]{Lemma}
\newtheorem*{lemma*}{Lemma}

\newtheorem{remark}[theorem]{Remark}

\numberwithin{equation}{section}



\input{header.tex}

\newcommand{\cref}[1]{Fig. \ref{#1}}

\begin{document}
\title[]{Homogenization of two-phase flow in porous media from Pore to Darcy Scale: A phase-field approach}
\date{}
\author{Stefan Metzger}
\address{Department of Mathematics, Friedrich--Alexander--Universit\"at Erlangen--N\"urnberg, Erlangen, Germany}
\email{stefan.metzger@fau.de}

\author{Peter Knabner}
\address{Department of Mathematics, Friedrich--Alexander--Universit\"at Erlangen--N\"urnberg, Erlangen, Germany}
\email{knabner@math.fau.de}

\selectlanguage{english}

\maketitle
\begin{abstract}
 We extend the two-scale expansion approach of periodic homogenization to include time scales and thus can tackle the full instationary Navier-Stokes-Cahn-Hilliard model at the pore scale as microscale. Time scale separation allows us to keep microscale dynamics, responsible e.g. for hysteresis, and arrive at a numerically tractable micro-macro model including coupled generalized Darcy's laws.
\end{abstract}
\section{Introduction}
The mathematical modeling and numerical simulation of multiphase flow (water-air, water-oil , \ldots) in porous media is of paramount importance in science and engineering. One widespread model combines mass conservation of each phase with a phase related generalized Darcy law (see e.g. Chapter 5.2 in \cite{Hornung1997})
\begin{subequations}\label{eq:twophasedarcy}
\begin{align}
\zeta\para{t}\trkla{S_\xi\rho_\xi}+\divx\tgkla{\rho_\xi\mathbf{v}_\xi}=f_\xi&&\text{in }\Omega\times(0,T)\,,\label{eq:darcy:sat}\\
\mathbf{v}_\xi=-\Lambda_{\xi}\trkla{S_w}\trkla{\nablax p_\xi +\rho_\xi g\mathbf{e}_3}&&\text{in }\Omega\times(0,T)\label{eq:darcy:vel}\,,
\end{align}
\end{subequations}
where the index $\xi\in\tgkla{w,n}$ is referring to the wetting phase $w$ or non-wetting phase $n$, respectively.
Here, $\rho_\xi$ denotes the mass density of phase $\xi$ and $\trkla{\mathbf{v}_\xi,p_\xi}$ are velocity and pressure of phase $\xi$.
By $\zeta$ and $g$, we denote the porosity and the gravitational acceleration, while $\Lambda_\xi$ denotes the mobility tensor of phase $\xi$, which is assumed to depend solely on the saturations $S_\xi$.
To close the system \eqref{eq:darcy:sat}-\eqref{eq:darcy:vel}, it is assumed that the saturations satisfy
\begin{align}\label{eq:twophase:saturations}
\Satw+\Satn=1
\end{align}
and that the pressure difference
\begin{align} \label{eq:pressure_difference}
  p_w-p_n=p_c\trkla{S_w},
\end{align}
the so-called capillary pressure, is uniquely given by the saturation $\Satw$.

This model is formulated in a spatial scale, which does not see the complex solid-fluid-fluid geometry anymore, i.e. on a macroscale (here also called Darcy scale), on which this geometry is only encoded through ``effective'' coefficients and relations (e.g. $\Lambda_{\xi}, p_{c}, \ldots$). It dates back to Richards and Muskat in the 1930ies, who extended intuitively Darcy's law for saturated one phase flow, and was also found experimentally in the 1860ies. This model has been applied very successfully over the decades to model and simulate multiphase flow processes (see e.g \cite{Helmig1997}). Nevertheless it has some major flaws becoming more and more visible and impeding in more complex applications. In particular in cannot explain hysteresis effects observed in measurements. Therefore, starting at least in the 1980ies, two related developments emerged: On the one hand side apporaches to justify this model by upscaling procedures starting from a continuum mechanics based models at the pore scale as the microscale, and on the other hand the proposal of modified or completely new models. For upscaling two procedures  are established: Volume averaging with its extension to an thermodynamically constrained averaging theory (TCAT) (see \cite{Gray_Miller2014} and the references therein), and (periodic) homogenization (see \cite{Hornung1997}). Both approaches start form a continuum mechanics based model on a geometry resolving the microscale, i.e. in the situation considered  here, where the microscale is the pore scale, there is a resolution of the solid geometry and the various fluid phases. Volume averaging tries to convert averages of differential quantities to differential quantities of averages, at the expense of further (typically boundary) terms modeling phase exchange. For these terms either conditions are identified which render them negligible or they are expressed by closure relations in macroscopic, i.e. averaged, quantities. In the TCAT version, entropy relations as side conditions try to reduce the ambiguity in selecting the closure relations. Therefore volume averaging is formal in various aspects, in particular applied to multiphase flows in a sharp interface formulation. Furthermore pure spatial averaging as it is done usually cannot take into account several temporal scales, which come in naturally considering instationary situations. \\ Periodic homogenization assumes spatial periodicity of a microstructure of the underlying domain $\Omega\subset\R^d$ with $d\in\tgkla{2,3}$. More precisely, for further usage we assume that  
it can be described using a reference cell $\RefDomain=(0,1)^d$. This reference cell consists of a solid part $\Solid$ and a fluid part $\Fluid$ with boundary $\dFluid$. Therefore, we may describe the complete fluidic domain $\Fluideps$ by $\bigcup_{z\in\mathds{Z}^d}\varepsilon\rkla{\Fluid+z}\cap\Omega$. Here, $\varepsilon<\!<\!1$ is the quotient of the typical length of the periodic microstructure and the macroscopic length scale (in this publication scaled to one). Similarly, we describe the solid matrix by $\Solideps=\bigcup_{z\in\mathds{Z}^d}\varepsilon\rkla{\Solid+z}\cap\Omega$ (cf. \cref{fig:domain}).
The boundary between the fluidic domain and the solid matrix will be denoted by $\dFluideps:=\overline{\Fluideps}\cap\overline{\Solideps}$.
 For the remainder of this paper, we assume that the fluidic domain $\Fluideps$ is open and  connected with a smooth boundary. For the $\varepsilon$ dependent solutions of the microscale model, convergence to the solution of a problem with no or at least less microstructure is aimed for, either in two-scale convergence (see e.g. \cite{Pavliotis_Stuart2008}) or by assuming a two-scale asymptotic expansion like \eqref{eq:expansion} and thus distinguishing between a ``slow'' ($\boldsymbol{x}$) and a fast ($ \boldsymbol{y} = \boldsymbol{x}/\varepsilon $) scale and aiming for determining equations of the leading terms. 

In Chapter 5.1 of \cite{Hornung1997} a slightly more general macroscopic two-phase, Darcy-like model, where the velocities are given by
\begin{subequations}\label{eq:darcy:general}
\begin{align}
\mb{v}_w=-\Lambda_w^w\lambda_w-\Lambda_w^n\lambda_n\,,\\
\mb{v}_n=-\Lambda_n^w\lambda_w-\Lambda_n^n\lambda_n\,
\end{align}
\end{subequations}
with $\lambda_\xi$ denoting the macroscopic forces acting on phase $\xi$, was derived by periodic homogenization.
This model is based on the following assumptions.
\begin{itemize}
\labitem{H1}{item:contactangle} $\Gamma$ intersects $\dFluideps$ with a wetting angle $\theta$.
\labitem{H2}{item:surfacetension} The surface tension $\hat{\tilde{\sigma}}$ is of order $\varepsilon$, $\hat{\tilde{\sigma}}=\varepsilon\tilde{\sigma}$ (to make the surface tension of order 1 at the level of the pore, because, then the curvature in a capillary tube is of order $\varepsilon$).
\labitem{H3}{item:jumpcondition} According to the Laplace relation, the jump of the normal component of the stress tensor across $\Gamma$ is given by $\jump{\trkla{p\mathds{1}-\boldsymbol{S}}\mathbf{n}}=\hat{\tilde{\sigma}}\trkla{R_1^{-1}+R_2^{-1}}\mathbf{n}$ where $\boldsymbol{S}$ is the viscous stress tensor, $R_1$ and $R_2$ are the main curvature radii, and $R_1^{-1}+R_2^{-1}$ is called the mean curvature. The tangential component is continuous.
\labitem{H4}{item:creepingflow} The flow in each phase is slow enough (creeping flow) to neglect inertial effects.
\labitem{H5}{item:stationaryinterface} The flow is assumed to be stationary. Then, the tangential component of the velocity is continuous across the interface $\Gamma$, and the normal component is zero.
\labitem{H6}{item:capillarypressure} There is a relation between the stress tensor jump through the interface and the saturation $\Sat_\xi$ in the cell.
\end{itemize}
From these conditions, in particular \ref{item:creepingflow}, \ref{item:stationaryinterface}, and \ref{item:capillarypressure} are critical: Assumption \ref{item:stationaryinterface} includes  a fixed interface position.
Therefore, detached droplets are not allowed to be transported but have to stay in place. A more detailed discussion of related results is given in Section~\ref{sec:comparison_models}.
Furthermore, \ref{item:creepingflow} implies the existence of a unique stationary state w.r.t. some faster time scale (cf. Remark \ref{rem:timescales}) and \ref{item:capillarypressure} prohibits dependency on the actual droplet topology. \\
In this paper we aim at applying an unconventional two-scale-expansion assumption to the full instationary Navier-Stokes equation with a phase-field description of the fluid-fluid interfaces. As indicated, it cannot be expected that we reach the ideal case of scale separation with purely macroscopic equations where the (periodic) pore geometry is only present  via cell problems determining the ``effective'' coefficients of the macroscopic equations. It is well-known for a for a long time (see e.g. \cite{Hornung1997}) that this is the case for stationary one phase flow, starting from the stationary Stokes equation on the microscale. First, when dealing with instationary problems, we have to relate temporal scales to the slow/fast spatial scales. We succeed in arriving at fast scale \eqref{eq:fast} and slow scale \eqref{eq:slow} equations, but still fully coupled with mixed derivatives in \eqref{eq:slow}. We can show energy estimates (Lemma~\ref{lem:micro:energy}, Lemma~\ref{lem:slow:energy}) for both equations showing that the system will evolve to a pseudo-stationary state (chemical potential $\mu_{0}$ independent of $\boldsymbol{y}$). This enables us, using Assumption \ref{item:creepingflow} for the first time, to recover a generalized Darcy's law in the form \eqref{eq:darcy:general} (cf. \eqref{eq:def:macrovel} or \eqref{eq:def:macrovel:2}). We eventually end up with a micro-macro model as described in Section~\ref{sec:numerics}, in which microscale cell problems are coupled in both ways at each point of the macroscopic domain. In Section~\ref{sec:numerics} we concentrate on microscale simulations to elucidate the impact to be expected from microscopic phenomena to macroscopic quantities. The efficient algorithmic treatment of such micro-macro models is still in its infancy (see e.g. \cite{Ray2019}) and out of the scope of this paper, but in particular due to their decoupling and adaptation potential they seem to be competitive with global extensions of the standard model: Also based on upscaling procedures, models have been proposed which either extend existing equations, e.g. making the relation \eqref{eq:pressure_difference} dynamic (see \cite{Hassan_Gray1993}) or by introducing new macroscopic quantities, i.e. interfacial area (see \cite{Niessner_Hassan2008}).

\begin{figure}
\includegraphics[width=\linewidth]{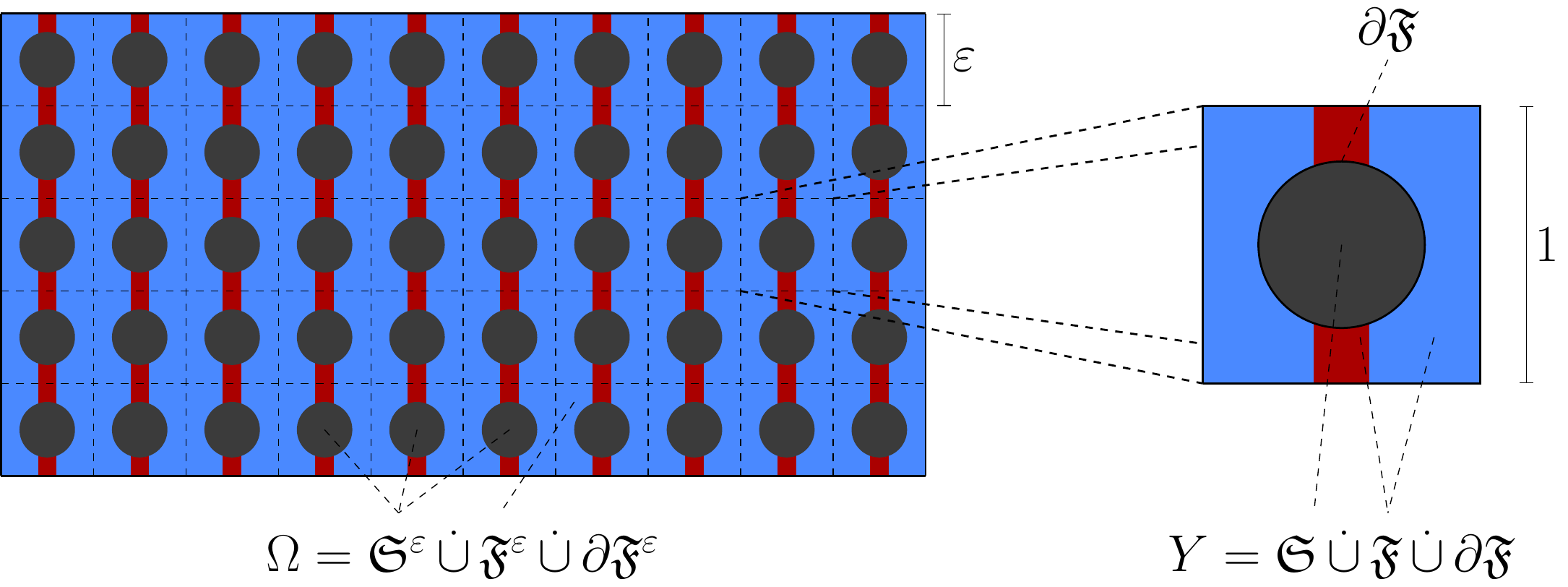}
\caption{Sketch of the periodic domain $\Domain$ comprising a periodic solid matrix $\Solideps$ and a fluid part $\Fluideps$ with boundary $\dFluideps$. $\Domain$ consist of periodic cells of length $\varepsilon$ which can be mapped onto the reference cell $\RefDomain$ of length 1 which analogously consists a solid part $\Solid$ and a fluid part $\Fluid$ separated by $\dFluid$.}
\label{fig:domain}
\end{figure}

\section{The governing equations}\label{sec:goveq}
As it is the aim of this publication to improve the understanding of the interactions between the different length scales and to obtain enhanced models, we refrain from using the standard (but inadequate) assumption that the flow is slow enough to neglect inertial effects (cf. \ref{item:creepingflow}) as long as possible.
Hence, it is not sufficient to describe the evolution of the fluid velocity by the stationary Stokes equations.
Instead, we shall consider the complete Navier--Stokes equations and include also effects on fast time scales.

Since we expect the effective parameters to rely heavily on the topology on the microscopic length scale, we describe the two-phase flow using a diffuse interface model (cf. \cite{Hohenberg1977,Lowengrub1998,Ding2007,AGG2012}).
The advantage of this approach is that topological changes can be handled without imposing artificial side conditions.

In such models, a so-called phase-field parameter $\phi\in\tekla{-1,+1}$ is used to describe the two fluid phases.
Throughout this manuscript, we will use the convention that the wetting phase is indicated by $\phi=+1$ and the non-wetting phase by $\phi=-1$.
Between the two pure phases, there exists a small transition region where $\phi$ varies smoothly.
Originally, phase-field models were derived for the equal mass density case (cf. \cite{Hohenberg1977}), but there are also several extensions to include the effects of different mass densities (see e.g. \cite{Gurtin1996,Lowengrub1998,Ding2007,AGG2012}).
In this publication, we use the model proposed by Abels, Garcke, and Grün in \cite{AGG2012}.
This model can be derived from general balance equations for mass and momentum by combining an energetic description of the system with Onsager's variational principle (cf. \cite{Onsager1931a,Onsager1931}).
It has the advantage of a solenoidal (volume averaged) fluid velocity field $\boldsymbol{u}$, while still satisfying an energy estimate for the total energy which consists of the kinetic energy $\mathcal{E}_{\text{kin}}:=\iFluideps \tfrac12\rho\trkla{\phi}\abs{\uflow}^2$ and the interfacial energy $\mathcal{E}_{\text{int}}:=\iFluideps \hat{\delta}\hat{\sigma}\tfrac12\abs{\nabla\phi}^2+\iFluideps \hat{\delta}^{-1}\hat{\sigma} W\trkla{\phi}$.
Here,
\begin{align}\label{eq:def:rho}
\rho\trkla{\phi}=\rhon\trkla{-\tfrac12\phi+\tfrac12}+\rhow\trkla{\tfrac12\phi+\tfrac12}\,,
\end{align}
interpolates between the mass density $\rhow:=\rho\trkla{+1}$ of the wetting phase and the mass density $\rhon:=\rho\trkla{-1}$ of the non-wetting phase.
As $\rho$ depends affine linearly on $\phi$, its derivative $\rho^\prime=\tfrac12\trkla{\rhow-\rhon}$ is constant.
$W\trkla{\phi}:=\tfrac14\trkla{1-\phi^2}^2$ is a polynomial double-well potential with minima in the pure phases $\phi=\pm1$.
The parameter $\hat{\delta}<\!\!<\varepsilon$ is related to the non-dimensional interface thickness which is assumed to be small with respect to the size of the periodic structure.
The parameter $\hat{\sigma}$ is related to the surface tension $\hat{\tilde{\sigma}}$ via $\hat{\sigma}=c_W\hat{\tilde{\sigma}}$ with a constant $c_W$ depending on the choice of the double-well potential.
For the polynomial double-well potential, we have $c_W=3/\trkla{2\sqrt{2}}$. 
The resulting microscopic model reads
\begin{subequations} \label{eq:micro:CH}
\begin{align}
&\para{t}\phi+\uflow\cdot\nabla\phi=\div\tgkla{\hat{M}\nabla\mu} &&\text{ in }\Fluideps\,,\label{eq:micro:CH:1}\\
&\mu=\hat{\sigma}\hat{\delta}^{-1}W^\prime\trkla{\phi}-\hat{\sigma}\hat{\delta}\Delta\phi &&\text{ in }\Fluideps\,,\label{eq:micro:CH:2}\\
&\nabla\mu\cdot\mathbf{n}=0&&\text{ on }\dFluideps\,,\label{eq:micro:CH:3}\\
&\nabla\phi\cdot\mathbf{n}=0 &&\text{ on }\dFluideps\,,\label{eq:micro:CH:4}
\end{align}
\end{subequations}
\begin{subequations}
\label{eq:micro:S}
\begin{multline}
\rho\trkla{\phi}\para{t}\uflow +\rho\trkla{\phi}\trkla{\uflow\cdot\nabla}\uflow-\trkla{\rho^\prime\hat{M}\nabla\mu\cdot\nabla}\uflow-\div\tgkla{2\hat{\eta}\trkla{\phi}\mathbf{D}\uflow}+\nabla p\\=-\phi\nabla\mu-\rho\trkla{\phi} g\mathbf{e}_3\qquad\qquad\text{ in } \Fluideps\,,\qquad\label{eq:micro:S:1}
\end{multline}
\begin{align}
&\div\uflow=0&\text{ in } \Fluideps\,,\label{eq:micro:S:2}\\
&~\uflow=0&\text{ on }\dFluideps\label{eq:micro:S:3}\,,
\end{align}
\end{subequations}
with suitable boundary conditions on $\partial\Omega$.
Here, $\mu$ denotes the chemical potential, i.e. the variation of the interfacial energy $\mathcal{E}_{\text{int}}$ with respect to the phase-field parameter $\phi$.
$\hat{M}$ denotes the mobility which is assumed to be constant and to scale linearly with $\hat{\delta}$, i.e. we have $\hat{M}=\hat{\delta}M$ with some constant $M>0$ (see also \cite{AGG2012} for other scalings).
$\mb{D}\uflow$ denotes the symmetrized gradient $\tfrac12\trkla{\nabla\uflow+\nabla^T\uflow}$, $g$ denotes the non-dimensional gravitational acceleration acting in $-\mb{e}_3$-direction and
\begin{align}\label{eq:def:eta}
\hat{\eta}\trkla{\phi}=\hatetan\trkla{-\tfrac12\phi+\tfrac12}+\hatetaw\trkla{\tfrac12\phi+\tfrac12}
\end{align}
interpolates between the viscosities $\etaw:=\eta\trkla{+1}$ and $\etan:=\eta\trkla{-1}$ of the wetting and non-wetting phase.

The boundary condition \eqref{eq:micro:CH:3} describes the impenetrability of the solid matrix and \eqref{eq:micro:CH:4} prescribes a contact angle of $90^\circ$.
The latter can be modified to allow for other contact angles (cf. \cite{Qian2006}).
However, for reasons of simplicity, we will stick to condition \eqref{eq:micro:CH:4} and the contact angle of $90^\circ$, but continue to speak of a wetting and a non-wetting phase.
The additional term $\trkla{\trkla{\rho^\prime\hat{M}\nabla\mu}\cdot\nabla}\uflow$ guarantees that mass density and momentum $\trkla{\rho\uflow}$ are transported with the same velocity (see \cite{AGG2012}), which is a fundamental prerequisite for the energy estimate formulated below.

Under the assumption of a closed system, i.e. there is no flux across the boundary $\partial\Omega$ and \eqref{eq:micro:CH:3}, \eqref{eq:micro:CH:4}, and \eqref{eq:micro:S:3} hold true on $\trkla{\partial\Omega\cap\overline{\Fluideps}}\cup\dFluideps$, it is possible to show a formal energy estimate. In particular, testing \eqref{eq:micro:CH:1} by $\mu+\tfrac12\rho^\prime\abs{\uflow}^2$, \eqref{eq:micro:CH:2} by $\para{t}\phi$, and \eqref{eq:micro:S:1} by $\uflow$, we obtain
\begin{multline}\label{eq:originalenergy}
\para{t}\iFluideps \tfrac12\rho\trkla{\phi}\abs{\uflow}^2+\para{t}\iFluideps\hat{\delta}\hat{\sigma}\tfrac12\abs{\nabla\phi}^2+\para{t}\iFluideps \hat{\delta}^{-1}\hat{\sigma} W\trkla{\phi}\\
+\iFluideps 2\hat{\eta}\trkla{\phi}\abs{\mb{D}\uflow}^2+\iFluideps\hat{M}\abs{\nabla\mu}^2=-\iFluideps\rho\trkla{\phi}g\mb{e}_3\cdot\uflow\,,
\end{multline}
i.e. the change in the total energy is given by the dissipation $\iFluideps 2\hat{\eta}\trkla{\phi}\abs{\mb{D}\uflow}^2+\iFluideps\hat{M}\abs{\nabla\mu}^2$ and the gravitational force.
The above identity is an important physical property of the model.
Hence, we will validate our results by showing that \eqref{eq:originalenergy} can be recovered after separating the scales (cf. Sections \ref{subsec:micro} and \ref{subsec:slow}).

\section{The homogenized equations}
The model introduced in Section \ref{sec:goveq} describes the two-phase flow on the microscopic level and consequently requires a fine resolution of the complete domain $\Omega$.
In this section, we shall identify the important spatial and temporal length scales.
By separating these scales, we derive a simplified model which allows us to describe the fast components of the system's evolution by equations which are local in space, and global equations that have to capture only the slower components.
After discussing the properties of the arising equations in Section \ref{subsec:micro} and Section \ref{subsec:slow}, we utilize Assumption \ref{item:creepingflow} in Section \ref{subsec:macro} to obtain purely macroscopic equations which are similar to the two-phase Darcy's law \eqref{eq:twophasedarcy}.
\subsection{Separation of scales}
The considered problem comprises two different small parameters: The ratio between the microscopic and macroscopic length scales $\varepsilon$ and the interface thickness $\hat{\delta}$ which is significantly smaller than the microscopic length scale.
As we will not compute the sharp interface limit $\hat{\delta}\searrow0$ of the model rigorously in this publication, we will assume $\hat{\delta}=\delta\varepsilon$ with some $\delta<\!\!<1$ to ensure $\hat{\delta}<\!\!<\varepsilon$ and consider only the case of vanishing $\varepsilon$.
For the sharp interface limit $\hat{\delta}\searrow0$, we refer to e.g. \cite{AGG2012}.
Using the separation of length scales, we introduce two different spatial coordinate systems with $\xflow\in \Omega$ denoting the macroscopic position and $\yflow\in\RefDomain$ representing the position in the reference cell.
Consequently, we expand the spatial derivatives via
\begin{align}\label{eq:expansion:space}
\nabla=\varepsilon^{-1}\nablay+\nablax\,.
\end{align}
Similarly, we separate two time scales.
The fast time scale $\taufast$ will capture the microscopic droplet movement and in particular the changes in the droplet topology, while the macroscopic changes in saturation will be considered on a slow time scale $\tauslow$.
Using these time scales, we expand the time derivatives via
\begin{align}\label{eq:expansion:time}
\para{t}=\varepsilon^{-1}\dtaufast+\dtauslow\,,
\end{align}
i.e. the two time scales are related via the same factor $\varepsilon$ as the spatial scales (see also Remark \ref{rem:timescales} below).
We expand the solutions in terms of $\varepsilon$.
In particular, we assume
\begin{subequations}
\label{eq:expansion}
\begin{align}
\uflow&= \uflow_0 +\varepsilon^2\uflow_1+\mathcal{O}\trkla{\varepsilon^3}\,,\label{eq:expansion:u}\\
\phi&=\phi_0+\varepsilon^2\phi_1+\mathcal{O}\trkla{\varepsilon^3}\,,\label{eq:expansion:phi}
\end{align}
i.e. we use expansions in terms of $\varepsilon^2$ for velocity and phase-field parameter.
For the chemical potential $\mu$ and the pressure $p$ different expansions are necessary.
In particular, special caution is appropriate when specifying the asymptotic expansion of $\mu$, as \eqref{eq:micro:CH:2} is merely a definition.
Consequently, choosing an expansion for $\mu$ which does not fit to this definition introduces additional artificial constraints.
In \cite{Schmuck13}, such an artificial constraint led to total blending of the immiscible fluids on the fine scale.
This results in a one-to-one relation between saturation and chemical potential based on a physically unstable microscopic state and prohibits hysteresis effects in the capillary pressure -- saturation curves (cf. \eqref{eq:pressurejump:pure}).
Plugging \eqref{eq:expansion:phi} into the right-hand side of \eqref{eq:micro:CH:2} shows that the asymptotic expansion of $\mu$ has to be chosen as
\begin{align}\label{eq:expansion:mu}
\mu&=\mu_{0}+\varepsilon\mu_1+\mathcal{O}\trkla{\varepsilon^2}\,,
\end{align}
to be consistent with the \eqref{eq:expansion:phi}.
Concerning the expansion of the pressure, we recall that $p$ serves as a Lagrange multiplier in \eqref{eq:micro:S:1} for the constraint \eqref{eq:micro:S:2}.
Therefore, the expansion for $p$ has to be chosen such that the equations arising from the different orders of \eqref{eq:micro:S:1} are always well-posed.
In our case, this is true for
\begin{align}\label{eq:expansion:p}
p&=p_0+\varepsilon p_1+\mathcal{O}\trkla{\varepsilon^2}\,.
\end{align}
\end{subequations}
\\
Adapting the scaling in \cite[Chapter 5]{Hornung1997}, we assume $\hat{\eta}\trkla{\phi}=\varepsilon^2\eta\trkla{\phi}$ and $\hat{\sigma}=\varepsilon\sigma$ (see also \ref{item:surfacetension}).
We substitute the aforementioned expansions into \eqref{eq:micro:CH} and \eqref{eq:micro:S} and sort the resulting terms by powers of $\varepsilon$.
At the leading order $\mathcal{O}\trkla{\varepsilon^{-1}}$, we obtain from \eqref{eq:micro:S:1} and \eqref{eq:micro:S:2}
\begin{subequations}\label{eq:fast}
\begin{align}\label{eq:fast:Stokes:1}
\rho\trkla{\phi_0}\dtaufast\uflow_0+\rho\trkla{\phi_0}\trkla{\uflow_0\cdot\nablay}\uflow_0 &-\trkla{\rho^\prime M\nablay\mu_0\cdot\nablay}\uflow_0+\nablay p_0=-\phi_0\nablay\mu_0\,,\\
\divy\uflow_0&=0\label{eq:fast:Stokes:2}
\end{align}
for $\yflow\in\Fluid$ and $\xflow\in\Omega$. Additionally, we obtain from \eqref{eq:micro:S:3}
\begin{align}
\uflow_0=0\label{eq:fast:boundaryu}
\end{align}
for $\yflow\in\dFluid$ and $\xflow\in\Omega$.
Noting $W^\prime\trkla{\phi}=W^\prime\trkla{\phi_0}+\mathcal{O}\trkla{\varepsilon^2}$ and collecting the terms of leading order $\mathcal{O}\trkla{\varepsilon^{-1}}$ in \eqref{eq:micro:CH:1} and $\mathcal{O}\trkla{\varepsilon^0}$ in \eqref{eq:micro:CH:2} provides
\label{eq:fast:CH}
\begin{align}
\dtaufast\phi_0&+\uflow_0\cdot\nablay\phi_0=\divy\tgkla{M\nablay\mu_{0}}\,,\label{eq:fast:CH:1}\\
\mu_{0}=&\sigma\delta^{-1}W^\prime\trkla{\phi_0}-\delta\sigma\Delta_\yflow\phi_0\label{eq:fast:CH:2}
\end{align}
for $\yflow\in\Fluid$ and $\xflow\in\Omega$. From \eqref{eq:micro:CH:3} and \eqref{eq:micro:CH:4}, we obtain the boundary conditions
\begin{align}
\nablay\mu_{0}\cdot\mathbf{n}&=0\,,\label{eq:fast:boundarymu}\\
\nablay\phi_0\cdot\mathbf{n}&=0\label{eq:fast:boundaryphi}
\end{align}
\end{subequations}
for $\yflow\in\dFluid$ and $\xflow\in\Omega$.\\
The equations \eqref{eq:fast} allow us to determine the microscopic behavior on the fast time scale.
As they include only spatial derivatives with respect to the cell coordinate $\yflow$, they can be solved in each macroscopic point $\xflow\in\Omega$ independently.
We shall discuss these equations in Section \ref{subsec:micro} in more detail.
In particular, we will prove the existence of a pseudo-stationary state with respect to the fast time scale and derive an expression for the capillary pressure.

The slow-scale equations arise from collecting the $\mathcal{O}\trkla{\varepsilon^0}$-terms in \eqref{eq:micro:S:1}, \eqref{eq:micro:S:2}, and \eqref{eq:micro:CH:1}, as well as the $\mathcal{O}\trkla{\varepsilon^1}$-terms in \eqref{eq:micro:CH:2}.
They read
\begin{subequations}\label{eq:slow}
\begin{multline}
\rho\trkla{\phi_0}\dtauslow\uflow_0+\rkla{\rkla{\rho\trkla{\phi_0}\uflow_0-\rho^\prime M\nablay\mu_0}\cdot\nablax}\uflow_0 -\rkla{\rho^\prime M\nablax\mu_0\cdot\nablay}\uflow_0\\
 -\rkla{\rho^\prime M\nablay\mu_1\cdot\nablay}\uflow_0
-\divy\gkla{2\eta\trkla{\phi_0}\mathbf{D_y}\uflow_0}+\nablay p_1+\nablax p_0\\
 = -\phi_0\nablay\mu_1-\phi_0\nablax\mu_0-\rho\trkla{\phi_0} g\mathbf{e}_3\,,\label{eq:slow:S:1}
\end{multline}
\begin{align}
\divx\uflow_0=0\,,\label{eq:slow:S:2}
\end{align}
\begin{align}
\dtauslow\phi_0 +\uflow_0\cdot\nablax\phi_0&=\divy\gkla{M\nablay\mu_1}+\divy\gkla{M\nablax\mu_0}+\divx\gkla{M\nablay\mu_0}\,,\label{eq:slow:CH:1}\\
\mu_1&=-\sigma\delta\divy\nablax\phi_0-\delta\sigma\divx\nablay\phi_0\,\label{eq:slow:CH:2}
\end{align}
with the additional boundary conditions
\begin{align}
\rkla{\nablay\mu_1+\nablax\mu_0}\cdot\mathbf{n}=0\,,\label{eq:slow:CH:bound:1}\\
\nablax\phi_0\cdot\mathbf{n}=0\,.\label{eq:slow:CH:bound:2}
\end{align}
\end{subequations}
In contrast to the fast scale equations \eqref{eq:fast}, the above equations include spatial derivatives with respect to $\xflow$ and $\yflow$ and therefore have to be solved on $\Omega\times\RefDomain$ simultaneously.
However, they do not include changes on the fast time scale anymore.
We will discuss the properties of these equations in Section \ref{subsec:slow}.
In Section \ref{subsec:macro}, we will combine them with Assumption \ref{item:creepingflow} to derive a purely macroscopic description of the slow scale evolution.

\begin{remark}[Time scales]\label{rem:timescales}
On the first glance, the choice of the two time scales $\taufast$ and $\tauslow$ and expansion \eqref{eq:expansion:time} might seem arbitrary.
However, recalling that velocity is the quotient of distance over time, it seems natural that the time scales for the microscopic transport and the macroscopic transport behave like the microscopic and macroscopic length scales.\\
Furthermore, this relation of the time scales can also be justified by mathematical considerations.
As we are interested in the macroscopic transport, it is obvious to consider at least the time scale $\tauslow$ which corresponds to the evolution of the macroscopic saturation.
On the other hand, omitting an explicit usage of the faster time scale $\taufast$ results in the disappearance of the time-derivatives in \eqref{eq:fast:CH:1} and \eqref{eq:fast:Stokes:1}.
This is equivalent to the assumption that the microscopic evolution of phase-field and velocity provides a unique stationary state, and therefore can be neglected.
Unfortunately, the considered partial differential equations are nonlinear and we can not expect the existence of such a unique stationary state, which is independent of initial data. Consequently, we need to consider the complete evolution from the initial data to the corresponding pseudo-stationary state (which we specify in Remark \ref{rem:pseudostationary}) and therefore require the time-derivatives in \eqref{eq:fast:CH:1} and \eqref{eq:fast:Stokes:1}.
\end{remark}

\subsection{The cell problem}\label{subsec:micro}
In this section, we will discuss the properties of the cell problem \eqref{eq:fast:Stokes:1}-\eqref{eq:fast:boundaryphi} and establish the existence of a pseudo-stationary state.
Thereby, the dissipation of free energy established in the following lemma will be a key ingredient.

\begin{lemma}[Microscopic energy estimate]\label{lem:micro:energy}
\!\!Let $\trkla{\phi_0,\mu_0,\uflow_0}$ be a solution to\eqref{eq:fast:Stokes:1}-\eqref{eq:fast:boundaryphi}.
Then, the leading order terms of the energy \eqref{eq:originalenergy} are not increasing on the fast time scale.
In particular, the following equality holds true for all $\xflow\in\Omega$:
\begin{multline}\label{eq:microscopicenergy}
\dtaufast\iFluid\tfrac12\rho\trkla{\phi_0}\abs{\uflow_0}^2+\dtaufast\iFluid\sigma\delta^{-1} W\trkla{\phi_0}+\dtaufast\iFluid\tfrac12\sigma\delta\abs{\nablay\phi_0}^2\\+\iFluid M\abs{\nablay\mu_0}^2=0\,.
\end{multline}
\end{lemma}
\begin{proof}
We start by testing \eqref{eq:fast:CH:1} by $\mu_0$ and \eqref{eq:fast:CH:2} by $\dtaufast\phi_0$, which yields
\begin{align}
\dtaufast\iFluid\sigma\delta^{-1} W\trkla{\phi_0} +\dtaufast\iFluid\tfrac12\sigma\delta\abs{\nablay\phi_0}^2+\iFluid\uflow_0\cdot\nabla\phi_0\mu_0+\iFluid M\abs{\nablay\mu_0}^2=0\,,
\end{align}
due to the boundary condition \eqref{eq:fast:boundarymu}.
In a second step, we test \eqref{eq:fast:Stokes:1} by $\uflow_0$ and obtain
\begin{multline}\label{eq:energyestimate:tmp:1}
\iFluid\tfrac12\rho\trkla{\phi_0}\dtaufast\abs{\uflow_0}^2 +\iFluid\rho\trkla{\phi_0}\rkla{\rkla{\uflow_0\cdot\nablay}\uflow_0}\cdot\uflow_0 - \iFluid \rkla{\rkla{\rho^\prime M\nablay\mu_0\cdot\nablay}\uflow_0}\uflow_0\\=\iFluid\phi_0\nablay\mu_0\cdot\uflow_0\,.
\end{multline}
Testing \eqref{eq:fast:CH:1} by $\tfrac12\abs{\uflow_0}^2\rho^\prime$ provides
\begin{align}\label{eq:energyestimate:tmp:2}
\iFluid\dtaufast\rho\trkla{\phi_0}\tfrac12\abs{\uflow_0}^2+\iFluid\tfrac12\abs{\uflow_0}^2\uflow\cdot\nablay\rho\trkla{\phi_0}+\iFluid\tfrac12 \rho^\prime M\nablay\mu_0\cdot\nablay\abs{\uflow_0}^2\,.
\end{align}
Straightforward computations show
\begin{multline}\label{eq:energyestimate:tmp:3}
\iFluid\rho\trkla{\phi_0}\rkla{\rkla{\uflow_0\cdot\nablay}\uflow_0}\cdot\uflow_0 +\iFluid\tfrac12\abs{\uflow_0}^2\uflow_0\cdot\nablay\rho\trkla{\phi_0}\\
=\iFluid\tfrac12\rho\trkla{\phi_0}\uflow_0\cdot\nablay\abs{\uflow_0}^2+\iFluid\tfrac12\abs{\uflow_0}^2\uflow_0\cdot\nablay\rho\trkla{\phi_0}=0\,.
\end{multline}
and
\begin{align}
-\iFluid\rho^\prime M\rkla{\rkla{\nablay\mu_0\cdot\nablay}\uflow_0}\cdot\uflow_0+\tfrac12\iFluid\rho^\prime M\nablay\mu_0\cdot\nablay\abs{\uflow_0}^2=0\,.
\end{align}
Therefore, adding \eqref{eq:energyestimate:tmp:1}-\eqref{eq:energyestimate:tmp:3} yields the result.
\end{proof}
\begin{remark}[Existence of a pseud-stationary state]\label{rem:pseudostationary}
Lemma \ref{lem:micro:energy} in particular tells us that the energy of the cell problem will decrease until $\mu_0$ is constant with respect to $\yflow$.
As the energy of the system is nonnegative, the system will evolve towards a pseudo-stationary state characterized by $\nablay\mu_0\equiv0$.
In this state, the shape of the droplet attains a (local) optimum, while we still allow for transportation.
This observation will be crucial for the numerical treatment of the model, as it provides a practical stopping criterion for the evolution of the microscopic problems.
\end{remark}
We conclude this section by deriving an expression for the capillary pressure.
Introducing the fluid pressure $p^f:=p+\mu\phi$ with the leading order term $p_0^f=p_0+\mu_0\phi_0$ and assuming that the system already reached a pseudo-stationary state, \eqref{eq:fast:Stokes:1} reads
\begin{align}\label{eq:micro:stationary}
\rho\trkla{\phi_0}\dtaufast\uflow_0+\rho\trkla{\phi_0}\trkla{\uflow_0\cdot\nablay}\uflow_0 +\nablay p^f_0=\mu_0\nablay\phi_0\,
\end{align}
with a constant $\mu_0$.
An integration perpendicular to the level set $\Gamma:=\tgkla{\phi_0\equiv0}$ provides an expression for the pressure difference between both fluids.
In particular, we obtain
\begin{align}
\delta F\trkla{\rho\trkla{\phi_0},\uflow_0}+\restr{p_0^f}{\phi_0=1}-\restr{p_0^f}{\phi_0=-1}=2\mu_0\label{eq:pressurejump}
\end{align}
with some $F\trkla{\rho\trkla{\phi_0},\uflow_0}$ stemming from the mean value integral of the first two terms in \eqref{eq:micro:stationary}.
Assuming that this $F$ is bounded independently of $\delta$, we neglect the first term for small $\delta$ and obtain
\begin{align}
\restr{p_0^f}{\phi_0=1}-\restr{p_0^f}{\phi_0=-1}=2\mu_0\label{eq:pressurejump:pure}
\end{align}
Recalling that $2\mu_0$ is a diffuse interface approximation of $\tilde{\sigma}\kappa$ (see e.g. Section 4 in \cite{AGG2012}), we notice that the capillary pressure depends on the surface tension $\tilde{\sigma}=c_W^{-1}\sigma$ and the mean curvature $\kappa$.
At this point, we want to emphasize that $\mu_0$ is not purely determined by the saturation, i.e. $\fint_{\RefDomain}\phi_0\dy$, but also depends on the interaction of the droplets with the solid matrix.
These interaction allow for hysteresis effects in the capillary pressure -- saturation curve.
For an illustration of this effect, we refer to Section \ref{sec:numerics} of this manuscript.
\subsection{The slow-scale problem}
\label{subsec:slow}
In this section, we investigate the properties of the slow-scale equations \eqref{eq:slow}.
Similar to Lemma \ref{lem:micro:energy}, we obtain an energy estimate with respect to the slow time-scale.
\begin{lemma}[Slow energy estimate]\label{lem:slow:energy}
A solution to \eqref{eq:slow:S:1}-\eqref{eq:slow:CH:bound:2} and \eqref{eq:fast:Stokes:1}-\eqref{eq:fast:boundaryphi} satisfies
\begin{multline}\label{eq:slow:energy}
\dtauslow\iDF\sigma\delta^{-1}W\trkla{\phi_0}+\dtauslow\iDF\tfrac12\sigma\delta\abs{\nablay\phi_0}^2 +\dtauslow\iDF\tfrac12\rho\trkla{\phi_0}\abs{\uflow_0}^2\\
+\dtaufast\iDF\sigma\delta\nablax\phi_0\cdot\nablay\phi_0-\idDF\sigma\delta\nablay\phi_0\cdot\mathbf{n}\dtaufast\phi_0
+\idDF\tfrac12\rho\trkla{\phi_0}\abs{\uflow_0}^2\uflow_0\cdot\mathbf{n}\\-\idDF\tfrac12\rho^\prime M\nablay\mu_0\cdot\mathbf{n}\abs{\uflow_0}^2
+2\iDF M\nablay\mu_0\cdot\nablay\mu_1+2\iDF M\nablax\mu_0\cdot\nablay\mu_0\\
-\idDF M\nablay\mu_0\cdot\mathbf{n}\mu_0
+\iDF2\eta\trkla{\phi_0}\abs{\mathbf{D_y}\uflow_0}^2\\
+\idDF\rkla{p_0+\phi_0\mu_0}\uflow_0\cdot\mathbf{n}+\iDF\rho\trkla{\phi_0}g\mathbf{e}_3\cdot\uflow_0=0\,.
\end{multline}
\end{lemma}
\begin{remark}
The identity stated in Lemma \ref{lem:slow:energy} includes boundary terms on $\partial\Omega$ describing the interactions with adjacent domains.
Under the assumption that the considered system is closed, i.e. \eqref{eq:fast:boundaryu}, \eqref{eq:fast:boundarymu}, and \eqref{eq:fast:boundaryphi} also hold true on $\partial\Omega\times\Fluid$, Lemma \ref{lem:micro:energy} and Lemma \ref{lem:slow:energy} provide the first two orders of the expansion of the energy balance equation \eqref{eq:originalenergy}.
Therefore, the derived multiscale model is still consistent with the original energetic description.
\end{remark}
\begin{proof}[Proof of Lemma \ref{lem:slow:energy}]
We start by testing \eqref{eq:slow:CH:1} by $\mu_0$ and $\rho^\prime\abs{\uflow_0}^2$ to obtain
\begin{subequations}
\begin{multline}\label{eq:slow:energy:CH1}
\iFluid\dtauslow\phi_0\mu_0 +\iFluid\divx\gkla{\uflow_0\phi_0}\mu_0 +\iFluid M\nablay\mu_0\cdot\nablay\mu_1 \\+\iFluid M\nablax\mu_0\cdot\nablay\mu_0 -\iFluid\divx\gkla{ M\nablay\mu_0}\mu_0=0\,,
\end{multline}
\begin{multline}\label{eq:slow:energy:CH2}
\iFluid\dtauslow\rho\trkla{\phi_0}\abs{\uflow_0}^2+\iFluid\uflow_0\cdot\nablax\rho\trkla{\phi_0}\abs{\uflow_0}^2+\iFluid\rho^\prime M\nablay\mu_1\cdot\nablay\abs{\uflow_0}^2 \\
+\iFluid \rho^\prime M\nablax\mu_0\cdot\nablay\abs{\uflow_0}^2 -\iFluid\divx\gkla{ \rho^\prime M \nablay\mu_0}\abs{\uflow_0}^2=0\,.
\end{multline}
\end{subequations}
Testing \eqref{eq:fast:CH:1} by $\mu_1$ provides
\begin{align}\label{eq:slow:energy:CH3}
\iFluid\dtaufast\phi_0\mu_1-\iFluid\uflow_0\phi_0\cdot\nablay\mu_1+\iFluid M\nablay\mu_0\cdot\nablay\mu_1=0\,.
\end{align}
Next, we test  \eqref{eq:slow:CH:2} by $\dtauslow\phi_0$ and \eqref{eq:fast:CH:2} by $\dtaufast\phi_0$ which yields
\begin{align}
\iFluid \mu_0\dtauslow\phi_0&=\iFluid \sigma\delta^{-1} \dtauslow W\trkla{\phi_0}+\iFluid\tfrac12\sigma\delta\dtauslow\abs{\nablay\phi_0}^2\,,\label{eq:slow:energy:mu0}\\
\iFluid\mu_1\dtaufast\phi_0&=-\iFluid\sigma\delta\divy\gkla{\nablax\phi_0}\dtaufast\phi_0-\iFluid\sigma\delta\divx\gkla{\nablay\phi_0}\dtaufast\phi_0\label{eq:slow:energy:mu1}\,.
\end{align}
Finally, we test \eqref{eq:slow:S:1} by $\uflow_0$ and obtain
\begin{multline}\label{eq:slow:energy:u}
\iFluid\rho\trkla{\phi_0}\tfrac12\dtauslow\abs{\uflow_0}^2+\iFluid \rkla{\nablax\uflow_0\cdot\rkla{\rho\trkla{\phi_0}\uflow_0-\rho^\prime M\nablay\mu_0}}\cdot\uflow_0\\
 -\iFluid \rkla{\nablay\uflow_0\cdot\trkla{\rho^\prime M\nablax\mu_0}}\uflow_0
 -\iFluid\rkla{\nablay\uflow_0\cdot\trkla{\rho^\prime M\nablay\mu_1}}\cdot\uflow_0 +\iFluid 2\eta\trkla{\phi_0}\abs{\mathbf{D_y}\uflow_0}^2 \\+\iFluid\uflow_0\cdot\nablax p_0
=-\iFluid\phi_0\nablay\mu_1\cdot\uflow_0-\iFluid\phi_0\nablax\mu_0\cdot\uflow_0 -\iFluid\rho\trkla{\phi_0} g\mathbf{e}_3\,.
\end{multline}
Summing \eqref{eq:slow:energy:CH1}, $\tfrac12\eqref{eq:slow:energy:CH2}$, \eqref{eq:slow:energy:CH3}, and \eqref{eq:slow:energy:u} and plugging in \eqref{eq:slow:energy:mu0} and \eqref{eq:slow:energy:mu1} yields
\begin{multline}
\dtauslow\iFluid\sigma\delta^{-1}W\trkla{\phi_0}+\dtauslow\iFluid\tfrac12\sigma\delta\abs{\nablay\phi_0}^2 +\dtauslow\iFluid\tfrac12\rho\trkla{\phi_0}\abs{\uflow_0}^2\\
-\iFluid\sigma\delta\divy\gkla{\nablax\phi_0}\dtaufast\phi_0-\iFluid\sigma\delta\divx\gkla{\nablay\phi_0}\dtaufast\phi_0\\
+\iFluid \tfrac12\rkla{\rho\trkla{\phi_0}\uflow_0-\rho^\prime M\nablay\mu_0}\cdot\nablax\abs{\uflow_0}^2+\iFluid\tfrac12\uflow_0\cdot\nablax\rho\trkla{\phi_0}\abs{\uflow_0}^2\\
-\iFluid\tfrac12\divx\gkla{\rho^\prime M\nablay\mu_0}\abs{\uflow_0}^2+2\iFluid M\nablay\mu_0\cdot\nablay\mu_1\\
+\iFluid M\nablax\mu_0\cdot\nablay\mu_0-\iFluid\divx\gkla{M\nablay\mu_0}\mu_0+\iFluid2\eta\trkla{\phi_0}\abs{\mathbf{D_y}\uflow_0}^2\\
+\iFluid \uflow_0\cdot\nablax p_0+\iFluid\nablax\rkla{\phi_0\mu_0}\cdot\uflow_0+\iFluid\rho\trkla{\phi_0}g\mathbf{e}_3\cdot\uflow_0=0\,.\label{eq:energy:slow:local}
\end{multline}
As \eqref{eq:energy:slow:local} is valid at every macroscopic point $\xflow\in\Omega$, an integration over the macroscopic domain yields the result.
\end{proof}
\subsection{Macroscopic equations}\label{subsec:macro}
In this section, we will use the slow-scale equations discussed in the last section to derive macroscopic models.
In the first part of this section, we derive a macroscopic model based on the previously used diffuse interface frame work.
In the second part, we formally derive an equivalent sharp interface version of the macroscopic equations.
This will help us later in the comparison with other established models.

\subsubsection{Macroscopic equations for the diffuse interface model}
According to the estimate derived in Lemma \ref{lem:micro:energy}, the system will evolve towards a pseudo-stationary state with $\nablay\mu_0=0$ on the fast time scale.
Therefore, we may safely assume $\nablay\mu_0=0$ when dealing only with the slow time scale.
Integrating \eqref{eq:slow:CH:1} over the fluidic domain $\Fluid$ and applying the boundary condition \eqref{eq:slow:CH:bound:1} provides the balance equation
\begin{align}\label{eq:macro:balance}
\dtauslow\iFluid \phi_0 + \divx\iFluid\uflow_0\phi_0=0\,.
\end{align}
Recalling that the wetting phase is indicated by $\phi=+1$, we define its saturation by $\Satw= \tfrac12\fint_{\Fluid}\phi_0+\tfrac12\in\tekla{0,1}$.
Analogously, the saturation of the non-wetting phase is given by $\Satn= -\tfrac12\fint_{\Fluid}\phi_0+\tfrac12\in\tekla{0,1}$.
Obviously, the saturations satisfy $\Satw+\Satn=1$ (cf. \eqref{eq:twophase:saturations}).
Noting $\abs{\RefDomain}=1$ and $\divx\uflow_0=0$, we may define the velocities $\vwd$ and $\vnd$ as
\begin{align}\label{eq:def:macrovel}
\vwd&=\iFluid\uflow_0\trkla{\tfrac12\phi_0+\tfrac12}=\iFluid\uflow_0\charwdelta\,, &\vnd&=\iFluid\uflow_0\trkla{-\tfrac12\phi_0+\tfrac12}=\iFluid\uflow_0\charndelta\,.
\end{align}
Here, $\charndelta:=\trkla{-\tfrac12\phi_0+\tfrac12}$ is a diffuse interface approximation of the characteristic function $\charn$ of $\Fluid_n$, the domain of the non-wetting phase associated with $\phi_0=-1$, and $\charwdelta:=\trkla{\tfrac12\phi_0+\tfrac12}$ is a diffuse interface approximation of the characteristic function $\charw$ of the domain $\Fluid_w$ of the wetting phase.
As both fluids are assumed to be incompressible, the fluid density is constant in the pure phases.
Consequently, \eqref{eq:macro:balance} provides
\begin{subequations}\label{eq:saturations}
\begin{align}
\zeta\dtauslow\trkla{\rhow\Satw}+\divx\gkla{\vwd\rhow}=0\,,\\
\zeta\dtauslow\trkla{\rhon\Satn}+\divx\gkla{\vnd\rhon}=0\,
\end{align}
\end{subequations}
with the porosity $\zeta:=\tabs{\Fluid}/\tabs{\RefDomain}=\tabs{\Fluid}$.
Hence, we have recovered \eqref{eq:darcy:sat} for the case absent source and sink terms.
Unfortunately, the evolution of $\uflow_0$ is prescribed by the nonlinear equations \eqref{eq:slow:S:1}, \eqref{eq:fast:Stokes:2}, and \eqref{eq:slow:S:2} and it is not possible to obtain a purely macroscopic, linear relation like \eqref{eq:darcy:vel} for $\vwd$ and $\vnd$ without imposing additional assumptions.

To simplify the equations even further, we will use Assumption \ref{item:creepingflow} which were used for the derivation of the two-phase Darcy law \eqref{eq:darcy:general}.
Assuming a known interface and applying \ref{item:creepingflow} and the linearity of $\rho$ and $\eta$, we obtain
\begin{multline}\label{eq:macro:1}
-\divy\gkla{2\etan\charndelta\mb{D_y}\uflow_0} -\divy\gkla{2\etaw\charwdelta\mb{D_y}\uflow_0}
+\nablay\ekla{\charndelta p_1}+\nablay\ekla{\charwdelta p_1}\\
=-\charndelta\nablax p_0-\charwdelta\nablax p_0 +\charndelta\nablax\mu_0-\charwdelta\nablax\mu_0
-\charndelta\rhon g\mb{e}_3 -\charwdelta\rhow g\mb{e}_3 -\phi_0\nablay\mu_1\,.
\end{multline}
Analyzing the right-hand side of \eqref{eq:macro:1}, we note that the velocity $\uflow_0$ is governed by forces acting only on the wetting fluid, forces acting on the non-wetting fluid, and the term $-\phi_0\nabla\mu_1$, which, as we will elucidate later, relates to the interactions at the fluid-fluid interface.
According to Lemma \ref{lem:micro:energy}, we can safely assume that $\mu_0$ is independent of $\yflow$.
Neglecting the inertial terms in \eqref{eq:micro:stationary}, we obtain
\begin{align}
\nablay p_0= -\phi_0\nablay\mu_0=0\,,
\end{align}
which makes also $p_0$ independent of the microscopic variable.
Therefore, we will use the auxiliary problems
\begin{subequations}\label{eq:auxn}
\begin{multline}
-\divy\gkla{2\etan\charndelta\mb{D_y}\mb{w}_k^{\delta,n}} -\divy\gkla{2\etaw\charwdelta\mb{D_y}\mb{w}_k^{\delta,n}}
\\+\nablay\ekla{\charndelta \pi_k^{\delta,n}}+\nablay\ekla{\charwdelta \pi_k^{\delta,n}}= \charndelta \mb{e}_k\,,
\end{multline}
\begin{align}
\divy\mb{w}_k^{\delta,n}=&0\,,
\end{align}
\end{subequations}
\begin{subequations}\label{eq:auxw}
\begin{multline}
-\divy\gkla{2\etan\charndelta\mb{D_y}\mb{w}_k^{\delta,w}} -\divy\gkla{2\etaw\charwdelta\mb{D_y}\mb{w}_k^{\delta,w}}
\\+\nablay\ekla{\charndelta \pi_k^{\delta,w}}+\nablay\ekla{\charwdelta \pi_k^{\delta,w}}= \charwdelta \mb{e}_k\,,
\end{multline}
\begin{align}
\divy\mb{w}_k^{\delta,w}=&0\,
\end{align}
\end{subequations}
with homogeneous Dirichlet boundary conditions on $\partial\Fluid$ and periodic boundary conditions on $\partial\RefDomain$ to compute the velocity components depending on $\rhow g\mb{e}_3$, $\rhon g \mb{e}_3$, $p_0$, and $\mu_0$.
In addition, we need to solve
\begin{subequations}\label{eq:auxint}
\begin{align}
-\divy\gkla{2\etan\charndelta\mb{D_y}\mb{w}_k^{\delta,\text{int}}} -\divy\gkla{2\etaw\charwdelta\mb{D_y}\mb{w}_k^{\delta,\text{int}}}
+\nablay\ekla{\charndelta \pi_k^{\delta,\text{int}}}&+\nablay\ekla{\charwdelta \pi_k^{\delta,\text{int}}}\\
=& -\phi_0\nabla\mu_1\nonumber\,,\\
\divy\mb{w}_k^{\delta,w}=&0\,,
\end{align}
\end{subequations}
to include the interfacial effects.
Therefore, we may write the fluid velocity $\uflow_0$ as
\begin{align}\label{eq:simplevel}
\uflow_0= \bs{K}^{\delta,n}\trkla{-\nablax p_0+\nablax\mu_0-\rhon g\mb{e}_3} + \bs{K}^{\delta,w}\trkla{-\nablax p_0-\nablax\mu_0-\rhow g\mb{e}_3} + \mb{w}^{\delta,\text{int}}\,,
\end{align}
with $\bs{K}^{\delta,n}:=\trkla{\bs{w}_1^{\delta,n},...,\bs{w}_d^{\delta,n}}$ and $\bs{K}^{\delta,w}:=\trkla{\bs{w}_1^{\delta,w},...,\bs{w}_d^{\delta,w}}$.
In order to obtain expressions for the macroscopic velocities $\vnd$ and $\vwd$ defined in \eqref{eq:def:macrovel}, we multiply \eqref{eq:simplevel} by $\charndelta$ (or $\charwdelta$, respectively) and integrate over $\Fluid$.
Hence, we have
\begin{subequations}\label{eq:def:macrovel:2}
\begin{multline}
\vnd=\bs{\Lambda}^{\delta,n}_n\trkla{-\nablax p_0+\nablax\mu_0-\rhon g\mb{e}_3} \\+ \bs{\Lambda}^{\delta,w}_n\trkla{-\nablax p_0-\nablax\mu_0-\rhow g\mb{e}_3} + \iFluid\charndelta\mb{w}^{\delta,\text{int}}\,,
\end{multline}
\begin{multline}
\vwd=\bs{\Lambda}^{\delta,n}_w\trkla{-\nablax p_0+\nablax\mu_0-\rhon g\mb{e}_3} \\+ \bs{\Lambda}^{\delta,w}_w\trkla{-\nablax p_0-\nablax\mu_0-\rhow g\mb{e}_3} + \iFluid\charwdelta\mb{w}^{\delta,\text{int}}\,,
\end{multline}
\end{subequations}
with $\bs{\Lambda}^{\delta,n}_n:=\iFluid\charndelta\bs{K}^{\delta,n}$, $\bs{\Lambda}^{\delta,w}_n:=\iFluid\charndelta\bs{K}^{\delta,w}$, $\bs{\Lambda}^{\delta,n}_w:=\iFluid\charwdelta\bs{K}^{\delta,n}$, and $\bs{\Lambda}^{\delta,w}_w:=\iFluid\charwdelta\bs{K}^{\delta,w}$.

\subsubsection{Macroscopic equations for sharp interface models}
In order to compare our newly derived model with already established ones like \eqref{eq:twophasedarcy} and \eqref{eq:darcy:general}, we transform it into a sharp interface model and explain the meaning of $-\phi_0\nabla\mu_1$.
First, we rewrite this term as $-\nablay\trkla{\phi_0\mu_1}+\mu_1\nablay\phi_0$.
The first term allows us to exchange $p_1$ by the fluid pressure $p_1^f=p_1+\phi_0\mu_1=\charwdelta\trkla{p_1+\mu_1}+\charndelta\trkla{p_1-\mu_1}$.
Concerning the second term, we write $\nablay\phi_0$ as $\abs{\nablay\phi_0}\mb{n}_n$ with the unit vector $\mb{n}_n:=\tfrac1{\tabs{\nablay\phi_0}}\nablay\phi_0$.
From the coarea formula (cf. \cite{Federer1959}), we obtain that for every function $f$
\begin{align}\label{eq:coarea:1}
\begin{split}
\iOmega f\trkla{\xflow}\abs{\nabla\phi\trkla{\xflow}}\text{d}\xflow&=\int_{-1}^{+1} \int_{\tgkla{\xflow\,:\,\phi\trkla{\xflow}=a}} f\trkla{\xflow} \text{d}H^{n-1}\trkla{\xflow} \text{d}a\\
&=2\fint_{-1}^{+1} \int_{\tgkla{\xflow\,:\,\phi\trkla{\xflow}=a}} f\trkla{\xflow} \text{d}H^{n-1}\trkla{\xflow} \text{d}a
\end{split}
\end{align}
Standard results for the sharp interface limit $\delta\searrow0$ (see e.g. \cite{EckGarckeKnabner_engl}) show that the $\phi$-profile perpendicular to the interface is approximately $\operatorname{tanh}\trkla{\tfrac{\text{sdist}\trkla{\xflow}}{\delta\sqrt{2}}}$, where $\text{sdist}\trkla{\xflow}$ denotes the signed distance of $\xflow$  to the interface given by $\Gamma:=\tgkla{\xflow\in\Omega\,:\,\phi_0\trkla{\xflow}=0}$.
Therefore, we obtain that the last line in \eqref{eq:coarea:1} converges towards
\begin{align}
\int_\Gamma f\trkla{\xflow}\text{d}H^{n-1}\trkla{\xflow}\,.
\end{align}
Therefore, $\tfrac12\tabs{\nablay\phi_0}$ can be seen as an approximation of the surface dirac function $\delta_\Gamma$ for the fluid-fluid interface $\Gamma$.
Recalling that $2\mu$ is a diffuse interface approximation of $\tilde{\sigma}\kappa$, we interpret $2\mu_1$ as the $\mathcal{O}\trkla{\eps}$-term of asymptotic expansion of the interfacial forces.
Replacing $\charndelta$, $\charwdelta$, and $\tabs{\nabla\phi}$ by $\charn$, $\charw$, and $\delta_\Gamma$ we obtain from \eqref{eq:macro:1}
\begin{multline}
\iFluid2\etan\charn\mb{D_y}\uflow_0:\mb{Dw}+\iFluid2\etaw\charw\mb{D_y}\uflow_0:\mb{Dw}-\iFluid\charn \trkla{p_1-\mu_1}\divy\mb{w}\\
-\iFluid\charw \trkla{p_1+\mu_1}\divy\mb{w}=-\iFluid\charn\nablax p_0\cdot\mb{w}-\iFluid\charw\nablax p_0\cdot\mb{w}
 +\iFluid\charn\nablax\mu_0\cdot\mb{w}\\-\iFluid\charw\nablax\mu_0\cdot\mb{w}
-\iFluid\charn\rhon g\mb{e}_3\cdot\mb{w} -\iFluid\charw\rhow g\mb{e}_3\cdot\mb{w} +\int_{\Gamma}2\mu_1\mb{n}_n\cdot\mb{w}\,,
\end{multline}
for any test function $\mb{w}$ vanishing on $\dFluid$.
This gives rise to the following set of equations
\begin{subequations}\label{eq:macro:stokes}
\begin{align}
-\divy\gkla{2\etan\mb{D_y}\uflow_0}+\nablay \ekla{p_1-\mu_1} &=-\nablax p_0 +\nablax\mu_0-\rhon g\mb{e}_3 &&\text{in } \Fluid_n\,,\label{eq:macro:stokes:n}\\
-\divy\gkla{2\etaw\mb{D_y}\uflow_0} +\nablay\ekla{p_1+\mu_1}&=-\nablax p_0-\nablax\mu_0-\rhow g\mb{e}_3 &&\text{in }\Fluid_w\,,\label{eq:macro:stokes:w}
\end{align}
where $\uflow_0$ is continuous across the fluid-fluid interface $\Gamma$ and satisfies the interface condition
\begin{align}
\jump{\trkla{p_1^f\mathds{1}-2\eta\mb{D_y}\uflow_0}\cdot\mb{n}_n}=2\mu_1\mb{n}_n\,.\label{eq:macro:stokes:jump}
\end{align}
\end{subequations}
Here, we kept the symbol $\mu_1$ for the $\mathcal{O}\trkla{\eps}$-term of asymptotic expansion of the interfacial forces in the sharp interface case.
Plugging the asymptotic expansions \eqref{eq:expansion:u}, \eqref{eq:expansion:phi}, \eqref{eq:expansion:mu}, and \eqref{eq:expansion:p} in the jump condition from \ref{item:jumpcondition}, we note that the terms of leading order $\mathcal{O}\trkla{1}$ correspond to \eqref{eq:pressurejump:pure}, while the order $\mathcal{O}\trkla{\varepsilon}$ provides \eqref{eq:macro:stokes:jump}.

Recalling the definition of $p_1^f$, the jump condition \eqref{eq:macro:stokes:jump} translates to
\begin{align}
\jump{\trkla{p_1\mathds{1}-2\eta\mb{D_y}\uflow_0}\cdot\mb{n}_n}=0\,.
\end{align}
This suggests that the cell problem \eqref{eq:auxint} has to be seen as an artifact of the diffuse interface description with vanishing influence on the velocity for $\delta\searrow0$ and therefore can be neglected.
Hence we may compute the effective macroscopic properties by using $\RefDomain$-periodic auxiliary problems similar to \eqref{eq:auxn} and \eqref{eq:auxw}.
These problems read
\begin{subequations}\label{eq:aux:nsharp}
\begin{align}
-\divy\gkla{2\etan\mb{D}\mb{w}^n_k}+\nablay\pi_k^n&=\mb{e}_k&&\text{in }\Fluid_w\,,\\
\divy\mb{w}^n_k &=0&&\text{in }\Fluid_w\,,\\
-\divy\gkla{2\etaw\mb{D}\mb{w}^n_k}+\nablay\pi_k^n&=0&&\text{in }\Fluid_n\,,\\
\divy\mb{w}^n_k &=0&&\text{in }\Fluid_n\,,\\
\mb{w}^n_k&=0&&\text{on }\dFluid\,,
\end{align}
with $\mb{w}^n_k$ being continuous across the interface and satisfying the jump condition
\begin{align}
\jump{\trkla{\pi^n_k\mathds{1}-2\eta\mb{D_y}\mb{w}^n_k}\cdot\mb{n}_n}=0
\end{align}
on the fluid-fluid interface $\Gamma$, and
\end{subequations}
\begin{subequations}\label{eq:aux:wsharp}
\begin{align}
-\divy\gkla{2\etan\mb{D}\mb{w}^w_k}+\nablay\pi_k^w&=0&&\text{in }\Fluid_w\,,\\
\divy\mb{w}^w_k &=0&&\text{in }\Fluid_w\,,\\
-\divy\gkla{2\etaw\mb{D}\mb{w}^w_k}+\nablay\pi_k^w&=\mb{e}_k&&\text{in }\Fluid_n\,,\\
\divy\mb{w}^w_k &=0&&\text{in }\Fluid_n\,,\\
\mb{w}^w_k&=0&&\text{on }\dFluid\,,
\end{align}
with $\mb{w}^w_k$ being continuous across the interface and satisfying the jump condition
\begin{align}
\jump{\trkla{\pi^w_k\mathds{1}-2\eta\mb{D_y}\mb{w}^w_k}\cdot\mb{n}_n}=0\,.
\end{align}
\end{subequations}
Similar to \eqref{eq:def:macrovel}, we obtain the macroscopic fluid velocities as
\begin{subequations}\label{eq:def:macrovel:2}
\begin{align}
\vn&=\bs{\Lambda}^{n}_n\trkla{-\nablax p_0+\nablax\mu_0-\rhon g\mb{e}_3} + \bs{\Lambda}^{w}_n\trkla{-\nablax p_0-\nablax\mu_0-\rhow g\mb{e}_3} \,,\\
\vw&=\bs{\Lambda}^{n}_w\trkla{-\nablax p_0+\nablax\mu_0-\rhon g\mb{e}_3} + \bs{\Lambda}^{w}_w\trkla{-\nablax p_0-\nablax\mu_0-\rhow g\mb{e}_3}\,,
\end{align}
\end{subequations}
with $\bs{\Lambda}^{n}_n:=\iFluidn\trkla{\bs{w}_1^n,...,\bs{w}_d^n}$, $\bs{\Lambda}^{w}_n:=\iFluidn\trkla{\bs{w}_1^w,...,\bs{w}_d^w}$, $\bs{\Lambda}^{n}_w:=\iFluidw\trkla{\bs{w}_1^n,...,\bs{w}_d^n}$, and $\bs{\Lambda}^{w}_w:=\iFluidw\trkla{\bs{w}_1^w,...,\bs{w}_d^w}$.
\section{Comparison with existing models}  \label{sec:comparison_models}
In this section, we will compare the upscaled models derived in the previous section with other models previously established in literature.
Thereby, we will consider the well-established generalized Darcy filtration law \eqref{eq:twophasedarcy} (see e.g. Chapter 5.2 in \cite{Hornung1997}) as well as recent developments based on diffuse interface descriptions.
\subsection{Periodic homogenization, fixed interface}
In Chapter 5 in \cite{Hornung1997}, several upscaled models for two-phase flow in porous media are presented.
Starting from a sharp interface cell problem which is similar to \eqref{eq:macro:stokes:n}-\eqref{eq:macro:stokes:w} and using the assumptions \ref{item:contactangle}-\ref{item:capillarypressure}, a generalized Darcy model with velocities \eqref{eq:darcy:general} is derived.
Although, the underlying equations seem to be similar, the resulting models differ in some important properties.\\
First of all, in our model the position of the interface is not treated as given, but computed by solving the fast scaled system \eqref{eq:fast}.
Consequently, the interface position depends not only on the saturation but also on the history of its evolution.
This allows us to explain the hysteresis in the relation between the capillary pressure and the saturation (see Section \ref{sec:numerics} for a more detailed investigation).\\
Secondly, we relax the assumption of a stationary interface (cf. \ref{item:stationaryinterface}) to the assumption of a known interface.
This allows us to consider transport of detached droplets in the reference cell by e.g. averaging the effective coefficients over time.
The downside is that we have to look for $\taufast$-periodic behavior to compute the macroscopic quantities based on different droplet topologies.
On the other hand, refraining from this assumption also has some remarkable advantages.
If the fluid-fluid interface is fixed, clogging phenomena may occur:
When fluid phase does not build a continuous channel through the cell $\RefDomain$, the averaged velocity of this phase has to vanish.
As forming a continuous channel through the cell requires a minimal saturation value and a specific microscopic fluid topology, this phenomenon occurs for small saturation values.
On the first glance, this might remind of the well-established concept of an irreducible saturation for which the permeability vanishes (cf. \cite{Hornung1997}).
However, this phenomenon is not based on adhesive forces gluing the wetting phase to the solid matrix, but on an artificial fixture of detached droplets.
By dropping this assumption, we allow for transport of detached droplets in $\RefDomain$, and thereby also for fluid transport for arbitrary saturations.
However, a droplet of the wetting phase may still stick to the solid matrix and resist the macroscopic forces.
Therefore, our model still allows for an irreducible saturation value.\\
The macroscopic model \eqref{eq:twophasedarcy} discussed e.g. in Chapter 5.2 in \cite{Hornung1997} is a further simplification, as it neglects the dependence of one phase on the forces acting on the other phase, which is caused by the continuity of the tangential component of the velocities across the the interface.

\subsection{Periodic homogenization, diffuse interface approach}
In the recent years, several upscaling approaches based on diffuse interface models for two-phase flow were published.
In \cite{DalyRoose15}, Daly and Roose combined the instationary Cahn--Hilliard equation with the stationary Stokes equations to describe the microscopic behavior of two-phase flow.
They identified three different time-scales -- namely a fast time-scale for the equilibration of the fluid-fluid interface, a medium time-scale for the transport on the fine spatial scale, and a large time-scale for the macroscopic transport.
They also derived relation \eqref{eq:pressurejump:pure} for the capillary pressure and provided numerical simulations showing hysteresis effects and discontinuities in the capillary pressure -- saturation relation.
Furthermore, they were also able to connect their cell problem to the one used in Chapter 5 in \cite{Hornung1997}.
However, in the upscaling process, they also used various assumptions that are not necessary in our approach.
For instance, they assumed that phase-field parameter $\phi_0$ can be written as the sum of the macroscopic saturation and a modulated part with average zero which depends only on the microscopic coordinate.
As a consequence, either the saturation is constant or $\phi_0$ will leave the physically meaningful interval $\tekla{-1,+1}$.
Using the standard expansion in terms of $\varepsilon$ for $\phi$ and $\uflow$ (in contrast to the expansion in terms $\varepsilon^2$ used in this publications), they obtained one phase-field equation comprising time-derivatives with respect to different time-scales.
Assuming that a time-independent mean value $\iFluid\phi_0$ implies a stationary state for $\phi_0$, they ended up with a cell problem based on a stationary Cahn--Hilliard equation.
As the stationary Cahn--Hilliard equation does not allow for a unique solution,
we keep the time-derivative in \eqref{eq:slow:CH:1} and omit the ambiguity of possible stationary states by connecting them to the initial conditions.\\

\subsection{Volume averaging, fixed interface}
An early example is \cite{Whitaker1986}. The author starts form stationary Stokes flow and under various conditions arrives at \eqref{eq:darcy:sat} and a variant of \eqref{eq:darcy:general}

\subsection{Volume averaging, diffuse interface approach}
In \cite{Chen2018} an instationary Navier--Stokes--Cahn--Hilliard system with slightly more complex boundary conditions, which we have omitted for technical reasons, but otherwise similar to our approach, is taken as starting point and volume averaging is applied.
One fundamental advantage of our approach seems to be the availability of the energy estimate.
This estimate not only provides the existence of a pseudo-stationary state, but may also serve as starting point for future analytical treatment of the model.\\
Furthermore, several simplifications are applied in \cite{Chen2018}, which partially seem to be critical:
The application of the mass conservation law in Section 3.5 in \cite{Chen2018} implies $\Delta\mu=0$, i.e. a pseudo-stationary state is assumed without the justification from the energy estimate.
Similar to our approach, they assume that the Reynolds number is small enough to allow for a reduction to the Stokes equations.
In contrast to our approach, only one time scale is considered in \cite{Chen2018}.
Consequently, the time step size in numerical simulations of the macroscopic behavior has to be chosen small enough to capture the rapid changes in the (averaged) interface curvature.
A second simplification is the negligence of interactions between the two phases ((3.50) to (3.51) there) resulting in Darcy relations of the form \eqref{eq:darcy:vel} instead of \eqref{eq:def:macrovel:2}.\\
Generally spoken, both approaches aim at different length scales, i.e. macroscopically still visible droplets in \cite{Chen2018} and macroscopically miscible fluids in our approach.

\section{Numerical treatment}\label{sec:numerics}
In this section, we will discuss an algorithm for the treatment of the newly derived model and employ numerical simulations to investigate the cell problem discussed in Section \ref{subsec:micro} with respect to capillary pressure -- saturation relation. \\
To obtain the evolution of the saturation, we have to solve the macroscopic equations \eqref{eq:saturations} and \eqref{eq:def:macrovel:2}.
These equations include effective parameters which must be obtained by solving cell problems in each macroscopic point $\xflow$.
Hence, we propose the following procedure for each large scale time step.
\begin{itemize}
\item[~]
\begin{itemize}
\item[\textbf{Step 0:}] Given from last time step: $\uflow_0\trkla{\xflow,\yflow}$, $\phi_0\trkla{\xflow,\yflow}$, $\Satw\trkla{\xflow}$, $\Satn\trkla{\xflow}$.
\item[\textbf{Step 1:}] Solve \eqref{eq:fast} on the reference cell with periodic boundary and for all $\xflow\in\Omega$:
\begin{itemize}
\item Adapt initial $\phi_0$ to match prescribed saturation.
\item Compute pseudo-stationary state.
\end{itemize}
\item[\textbf{Step 2:}] Compute effective parameters on the reference cell for all $\xflow\in\Omega$.
\begin{itemize}
\item Identify $\taufast$-periodicity interval $\mathcal{P}\trkla{\xflow}$.
\item Solve auxiliary problems \eqref{eq:auxn} and \eqref{eq:auxw} for all $\taufast\in \mathcal{P}\trkla{\mb{x}}$ to compute $\taufast$-dependent effective coefficients $\bs{\Lambda}^{\delta,n}_n$, $\bs{\Lambda}^{\delta,n}_w$, $\bs{\Lambda}^{\delta,w}_n$, and $\bs{\Lambda}^{\delta,w}_w$, depending on $\tau_{-1}$, and average over $\mathcal{P}(x)$.
\end{itemize}
\item[\textbf{Step 3:}] Update macroscopic quantities.
\begin{itemize}
\item Solve macroscopic equations \eqref{eq:saturations}, \eqref{eq:def:macrovel:2} to update $\Satw$, $\Satn$.
\end{itemize}
\end{itemize}
\end{itemize}
In each cell problem, we use the droplet topology from the last time step as initial condition, add a source or sink term to the Cahn--Hilliard equation to adjust $\iFluid\phi_0$ to the prescribed saturation, and search for a pseudo-stationary state of  \eqref{eq:fast} specified by $\nablay\mu_0\equiv0$.
Such a state exists due to Lemma \ref{lem:micro:energy}.\\
As the interface may not be stationary but transported through $\Fluid$, we have to identify periodic behavior with respect to the fast time-scale, and track the evolution of $\phi_0$ over the periodicity interval $\mathcal{P}\trkla{\xflow}$.
For all topologies in this interval, we solve the auxiliary problem \eqref{eq:auxn} and \eqref{eq:auxw} and average over $\mathcal{P}\trkla{\xflow}$ to obtain the effective parameters.
These computations can be done in parallel, as the different cell problems do not interact with each other.
In Step 3, we finally update the saturations by solving the upscaled equations.\\

In the remainder of this section, we will investigate the connection between the capillary pressure, the droplet topology, and the saturation.
Thereby, we will neglect the fluid velocity and compute the capillary pressure for various saturation values. Adding a source term to \eqref{eq:fast:CH:1}, allows us to control the mean value $\iFluid\phi_0$ by manipulating $\phi_0$ in the interfacial area.
As the stationary state of the system, which provides the capillary pressure via \eqref{eq:pressurejump:pure}, does not only depend on the mean value of $\phi_0$ (i.e. the saturation), but also on the initial conditions, we will proceed as follows.
We start with prescribed initial droplet configuration for a low saturation value of $0.0625$, increase the saturation step by step up to a maximal saturation of $0.9375$, before reducing it again to $0.0625$.
In this process, we will always use the last computed stationary state as initial configuration for the new simulation.

\subsection{Test case 1}
In our first test case, we consider a setup with a rather large porosity $\zeta=0.75$.
The reference cell $\RefDomain=\trkla{0,1}^2$ consists of the solid inclusion $\Solid=(0.25,0.75)^2$ and the fluidic domain $\Fluid=\RefDomain\setminus\Solid$.
Initially, we place a semicircular shaped droplet (center point at $(0.5,0.75)$) on top of the solid inclusion.
As we are searching for a stationary state for each saturation value and start with initial conditions that are not necessarily close to it, we expect to have time intervals, where the droplet topology changes rapidly, and time intervals, where the changes are hardly detectable.
To accommodate for this, we use an adaptive time discretization with time step sizes between $\tau_{\text{min}}=10^{-5}$ and $\tau_{\text{max}}=5.06\cdot10^{-2}$.
For the spatial discretization of the domain, we use a triangulation consisting of simplices with diameter $h\approx7.81\cdot 10^{-3}$.
The remaining parameters can be found in Table \ref{tab:tc1:param}.
\begin{table}[h]\center
\begin{tabular}{cccc||ccc}
$\zeta$&$M$ & $\sigma$ & $\delta$ & $h$ & $\tau_{\text{min}}$ & $\tau_{\text{max}}$\\
\hline
0.75&0.5 & 1 & 0.02 & $7.81\cdot 10^{-3}$ & $10^{-5}$ & $5.06\cdot 10^{-2}$
\end{tabular}
\caption{Parameters used in test case 1}
\label{tab:tc1:param}
\end{table}

\begin{figure}
\begin{minipage}[t]{0.3\textwidth}\center
\includegraphics[width=0.9\textwidth]{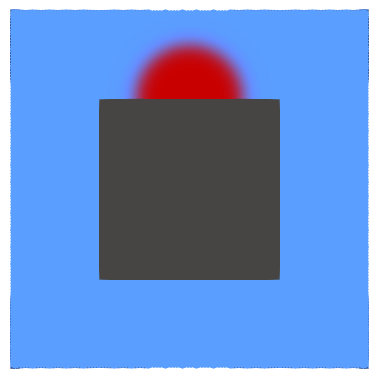}
\subcaption{$\Satw=0.0625$, wetting}
\label{fig:tc1:wet:0625}
\end{minipage}
\hfill
\begin{minipage}[t]{0.3\textwidth}\center
\includegraphics[width=0.9\textwidth]{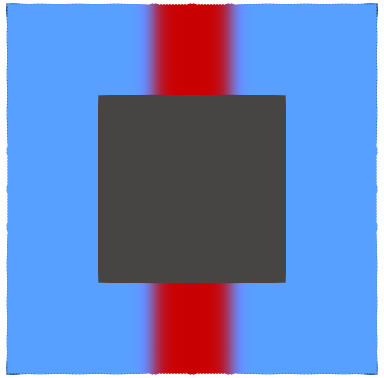}
\subcaption{$\Satw=0.137$, drainage}
\label{fig:tc1:drain:137}
\end{minipage}
\hfill
\begin{minipage}[t]{0.3\textwidth}\center
\includegraphics[width=0.9\textwidth]{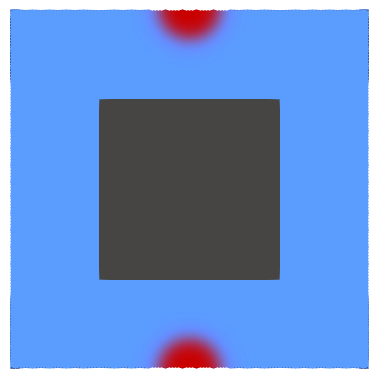}
\subcaption{$\Satw=0.0625$, drainage}
\label{fig:tc1:drain:0625}
\end{minipage}
\caption{Droplet topology for different saturation values.}
\end{figure}
\begin{figure}\center
\includegraphics[width=0.5\textwidth]{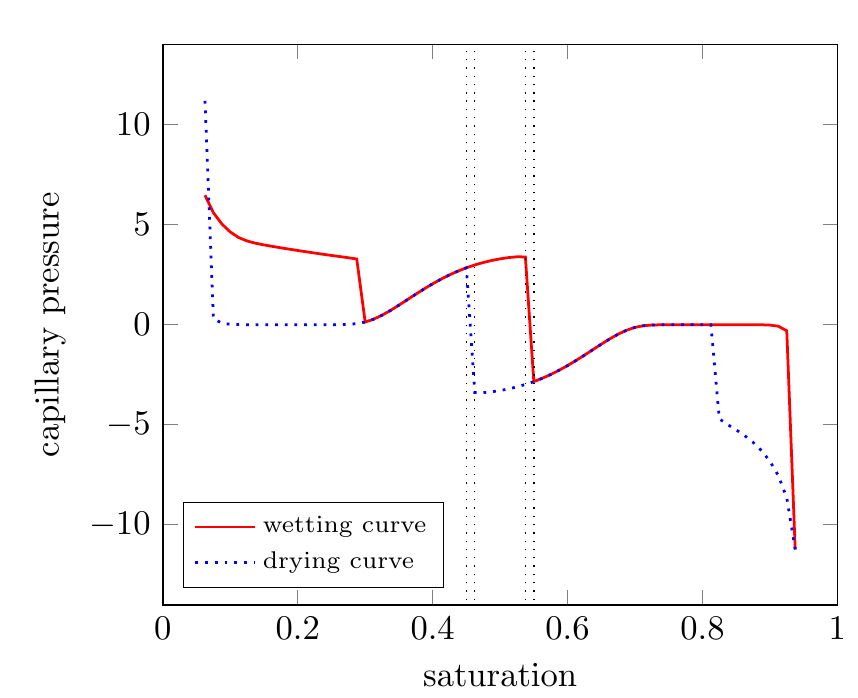}
\caption{Capillary pressure -- saturation relation for test case 1 in one cell.}
\label{fig:cap-sat:tc1}
\end{figure}

\begin{figure}[h]
\begin{minipage}[t]{0.38\textwidth}\center
\includegraphics[width=0.9\textwidth]{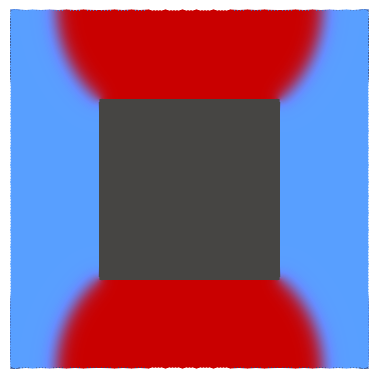}
\subcaption{$\Satw=0.45$, wetting}
\label{fig:tc1:wet:45}
\includegraphics[width=0.9\textwidth]{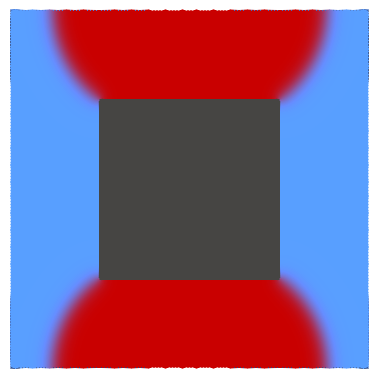}
\subcaption{$\Satw=0.4625$, wetting}
\label{fig:tc1:wet:4625}
\includegraphics[width=0.9\textwidth]{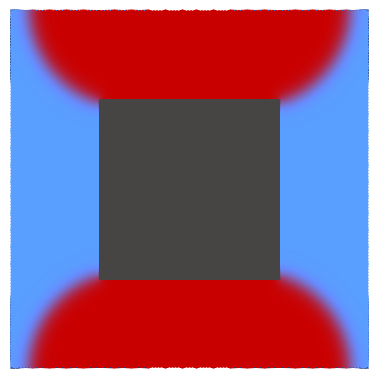}
\subcaption{$\Satw=0.5375$, wetting}
\label{fig:tc1:wet:5375}
\includegraphics[width=0.9\textwidth]{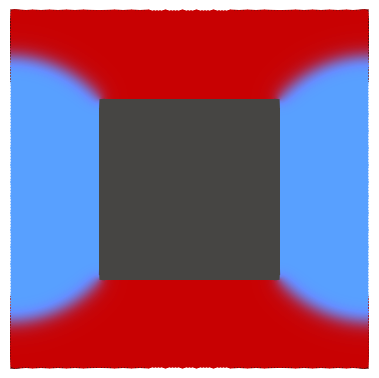}
\subcaption{$\Satw=0.55$, wetting}
\label{fig:tc1:wet:55}
\end{minipage}
\hfill
\begin{minipage}[t]{0.38\textwidth}\center
\includegraphics[width=0.9\textwidth]{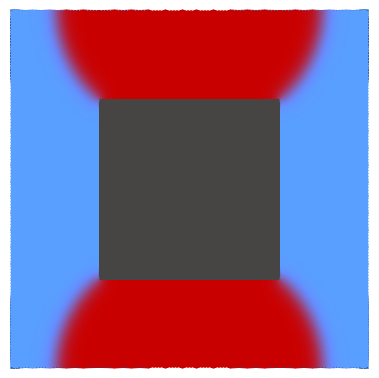}
\setcounter{subfigure}{7}
\subcaption{$\Satw=0.45$, drainage}
\label{fig:tc1:drain:45}
\includegraphics[width=0.9\textwidth]{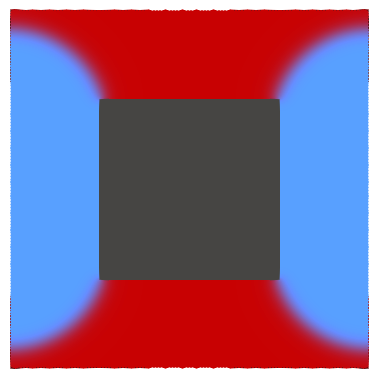}
\setcounter{subfigure}{6}
\subcaption{$\Satw=0.4625$, drainage}
\label{fig:tc1:drain:4625}
\includegraphics[width=0.9\textwidth]{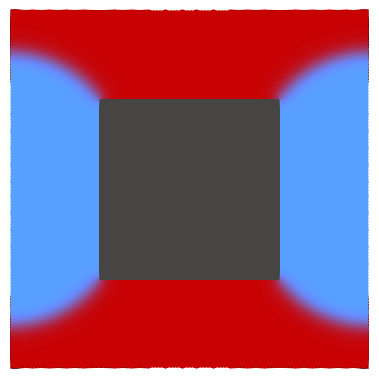}
\setcounter{subfigure}{5}
\subcaption{$\Satw=0.5375$, drainage}
\label{fig:tc1:drain:5375}
\includegraphics[width=0.9\textwidth]{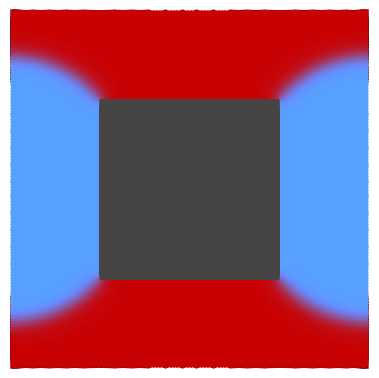}
\setcounter{subfigure}{4}
\subcaption{$\Satw=0.55$, drainage}
\label{fig:tc1:drain:55}
\end{minipage}
\center
\caption{Droplet evolution in test case 1.}
\label{fig:TC1:evo}
\end{figure}
Figure \ref{fig:tc1:wet:0625} shows the stationary state for the smallest considered saturation $\Sat_w=0.0625$.
Starting with this configuration, we increase the saturation gradually, wait until the system is in equilibrium, and measure the capillary pressure according to \eqref{eq:pressurejump:pure}.
The obtained curve is depicted in Figure \ref{fig:cap-sat:tc1}.
Thereby, the red, continuous line refers to the wetting curve associated with an increasing saturation and the blue dotted line refers to the drying curve for a decreasing saturation.
The vertical dotted, black lines indicate the saturation values corresponding the topologies depicted in Figure \ref{fig:TC1:evo}.\\
The curves plotted in Figure \ref{fig:cap-sat:tc1} exhibit some remarkable properties.
Looking at the wetting and drainage curves separately, we notice certain saturation values at which the capillary pressure seems to jump.
The first jump in the wetting curve occurs between $\Satw=0.2875$ and $\Satw=0.3$, where the droplet is large enough to touch and wet the next solid inclusion.
The next jump occurs between $\Satw=0.5375$ and $\Satw=0.55$.
Between these values, the wetting droplet (depicted in red in Fig. \ref{fig:TC1:evo}) reaches a critical size, where it touches and merges with the droplets of the adjacent periodicity cells (cf. Fig. \ref{fig:tc1:wet:5375}-\ref{fig:tc1:wet:55}) resulting in a massive change of the interface curvature.
The remaining pressure jumps can be explained in a similar way (cf. Fig. \ref{fig:tc1:drain:4625}-\ref{fig:tc1:drain:45}).\\
Secondly, the wetting and drying curves are not identical but show hysteresis effects.
This can be explained by the fact that the droplets need a certain size to touch and merge, but once they are connected, they stay connected and only break apart at a much lower saturation.
Consequently, the droplet topology does not only depend on the saturation but also on the history.
For illustration we refer to Fig. \ref{fig:tc1:wet:4625} and \ref{fig:tc1:drain:4625} and Fig. \ref{fig:tc1:wet:5375} and \ref{fig:tc1:drain:5375}.
Although, the topologies shown in the same row correspond to the same saturation value, the droplet topology is completely different.\\
One may also notice that the capillary pressure for very small saturations prescribed by the drainage curve is significantly higher than the one given by the wetting curve, which is not observed in real measurements.
The cause for this unexpected behavior is the assumption of $90^\circ$-contact angles resulting the creation of a small band connecting two adjacent solid inclusions (see Fig. \ref{fig:tc1:drain:137}).
With decreasing saturation, this band gets smaller and eventually breaks.
This rupture can occur at any point, but due to numerical artifacts, the rupture in our simulation occurred at the liquid-solid contact area, resulting in a detached droplet (see Fig. \ref{fig:tc1:drain:0625}).
As the curvature of a detached droplet is higher than the one of a droplet wetting the solid inclusion, the capillary pressure is higher.
This effect can be omitted by changing the prescribed contact angles.

\subsection{Test case 2}
In the second scenario, we look at a setting with a smaller porosity $\zeta=0.25$.
The reference cell $\RefDomain=(0,1)^2$ is a solid matrix $\Solid$ containing a vug and two channels connecting adjacent vugs. 
In this way we mimic typical geometries of pore network models (see \cite{Joekar_atalii2008}).
To be more precise, we have $\Fluid=\rkla{0.375,0.625}^2\cup \trkla{0.375,0.625}\times\trkla{0.4375,0.5625}\cup\trkla{0.4375,0.5625}\times\trkla{0.375,0.625}$ and $\Solid=\RefDomain\setminus\Fluid$ and place a droplet of the wetting fluid (depicted in red) in one of the channels (see Fig. \ref{fig:tc2:init}).
In this simulation, we used adaptive triangulation consisting of simplices with a maximal diameter of $h_{\text{max}}=1.56\cdot10^{-2}$ in the bulk areas and $h_{\text{min}}=2.76\cdot 10^{-3}$ in the interfacial area.
Similar to the first test case, an adaptive choice of the time step sizes between $\tau_{\text{max}}=4.87\cdot 10^{-2}$ and $\tau_{\text{min}}=2.93\cdot10^{-7}$ allows to accomodate for fast and slow changes.
For the remaining parameters, we refer to Table \ref{tab:tc2:param}.

\begin{table}[h]\center
\begin{tabular}{cccc||cccc}
$\zeta$&$M$ & $\sigma$ & $\delta$ & $h_{\text{min}}$ & $h_{\text{max}}$ & $\tau_{\text{min}}$ & $\tau_{\text{max}}$\\
\hline
0.25&0.1 & 1 & 0.005 & $2.76\cdot 10^{-3}$ & $1.56\cdot10^{-2}$ & $2.93\cdot10^{-7}$ & $4.87\cdot 10^{-2}$
\end{tabular}
\caption{Parameters used in test case 2}
\label{tab:tc2:param}
\end{table}

\begin{figure}\center
\includegraphics[width=0.3\textwidth]{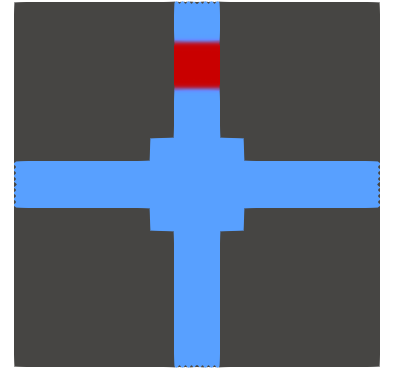}
\caption{Droplet configuration for $\Satw=0.0625$.}
\label{fig:tc2:init}
\end{figure}

\begin{figure}\center
\includegraphics[width=0.5\textwidth]{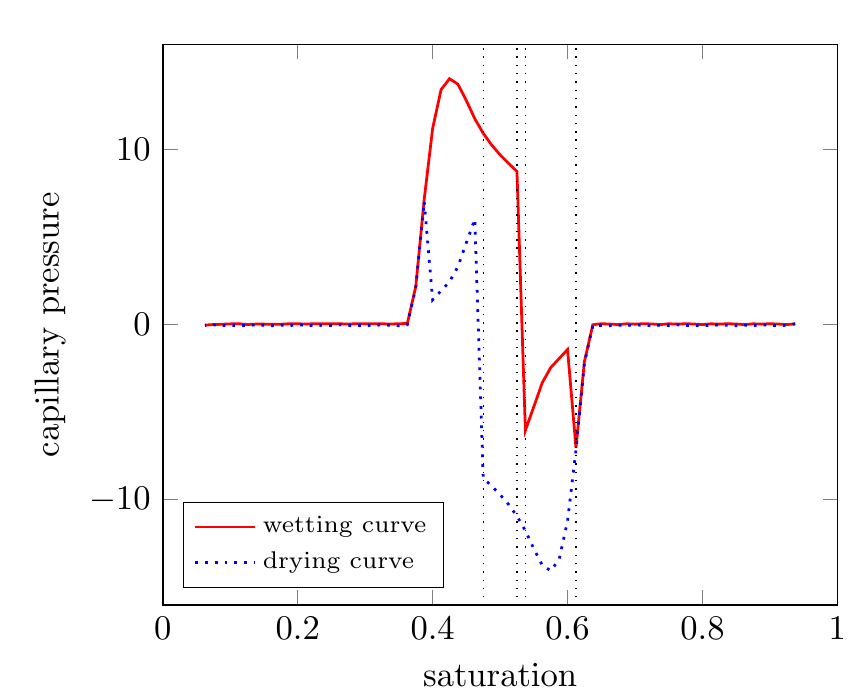}
\caption{Capillary pressure -- saturation relation for test case 2 in one cell.}
\label{fig:cap-sat:tc2}
\end{figure}
\begin{figure}[h]
\begin{minipage}[t]{0.38\textwidth}\center
\includegraphics[width=0.9\textwidth]{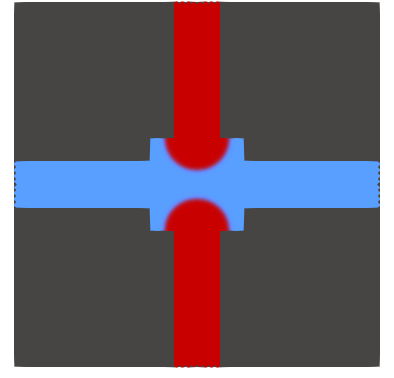}
\subcaption{$\Satw=0.475$, wetting}
\label{fig:tc2:wet:475}
\includegraphics[width=0.9\textwidth]{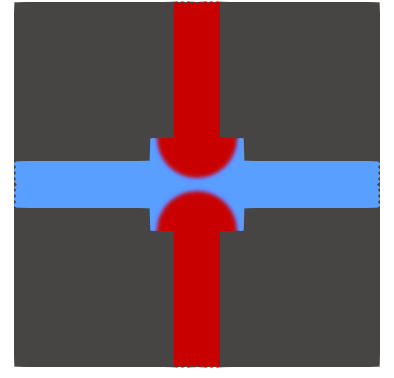}
\subcaption{$\Satw=0.525$, wetting}
\label{fig:tc2:wet:525}
\includegraphics[width=0.9\textwidth]{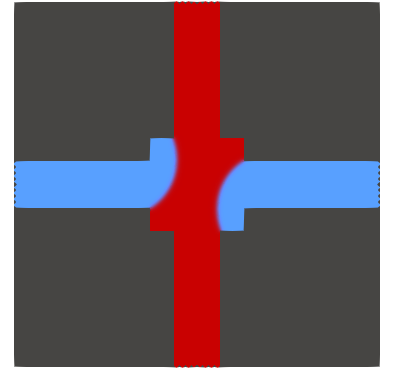}
\subcaption{$\Satw=0.5375$, wetting}
\label{fig:tc2:wet:5375}
\includegraphics[width=0.9\textwidth]{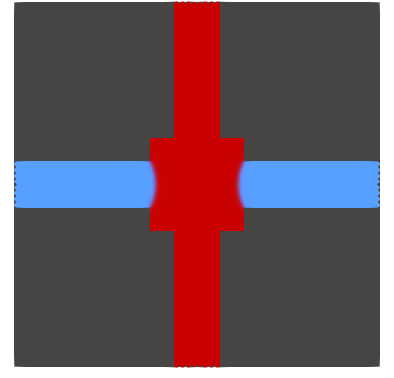}
\subcaption{$\Satw=0.6125$, wetting}
\label{fig:tc2:wet:6125}
\end{minipage}
\hfill
\begin{minipage}[t]{0.38\textwidth}\center
\includegraphics[width=0.9\textwidth]{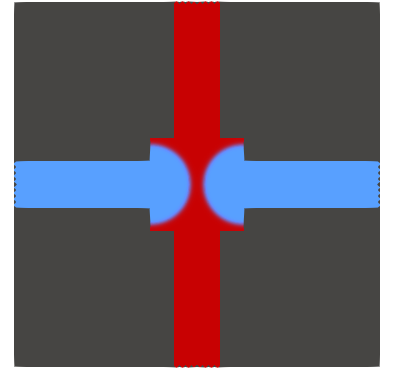}
\setcounter{subfigure}{7}
\subcaption{$\Satw=0.475$, drainage}
\label{fig:tc2:drain:475}
\includegraphics[width=0.9\textwidth]{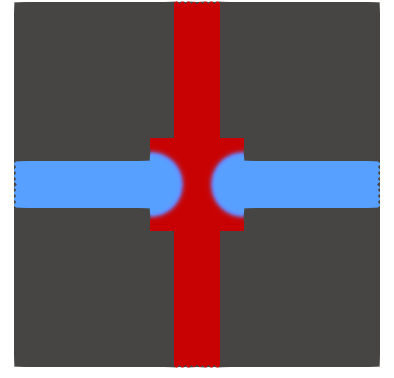}
\setcounter{subfigure}{6}
\subcaption{$\Satw=0.525$, drainage}
\label{fig:tc2:drain:525}
\includegraphics[width=0.9\textwidth]{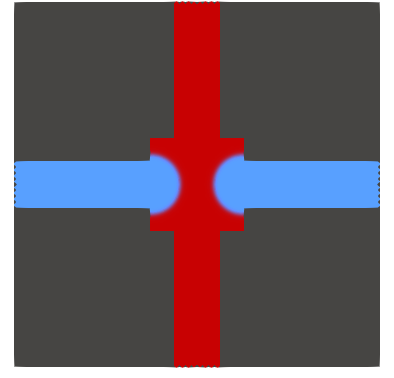}
\setcounter{subfigure}{5}
\subcaption{$\Satw=0.5375$, drainage}
\label{fig:tc2:drain:5375}
\includegraphics[width=0.9\textwidth]{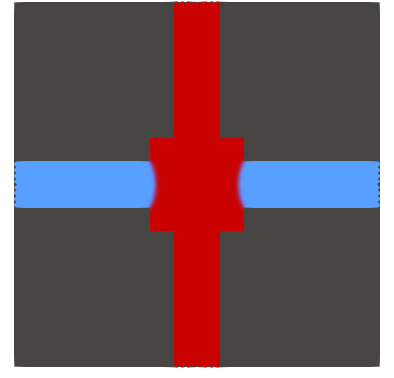}
\setcounter{subfigure}{4}
\subcaption{$\Satw=0.6125$, drainage}
\label{fig:tc2:drain:6125}
\end{minipage}
\center
\caption{Droplet evolution in test case 2.}
\label{fig:TC2:evo}
\end{figure}
Again, we are interested in the relation between capillary pressure and saturation.
Results from our simulations are plotted in Fig. \ref{fig:cap-sat:tc2}.
It is noticeable that the capillary pressure is zero for most saturations.
As the droplet is initially in one of the small channels and is assumed to attach to the wall with a contact angle of $90^\circ$, the arising fluid-fluid interface has zero curvature (cf. Fig. \ref{fig:tc2:init}).
When the droplet grows, it starts to fill the channel but does not change the interface curvature until the wetting phase starts flooding into the vug (see Fig \ref{fig:tc2:wet:475}).
At his point, the capillary pressure starts to rise with increasing curvature (cf. Fig. \ref{fig:tc2:wet:475} and Fig. \ref{fig:tc2:wet:525}) until adjacent droplets connect, which leads to an abrupt change in curvature and capillary pressure (see Fig. \ref{fig:tc2:wet:5375}).\\
Similar to our first test case, we observe hysteresis effects in the capillary pressure-saturation relation.
As depicted in \cref{fig:TC2:evo}, the critical saturation needed to connect the wetting phase in vertical direction is higher than the saturation needed to keep an already established connection.

\section{Discussion and Outlook}
The aim of this paper is to go significantly beyond the known results deriving a (mostly) macroscale  model for two-phase flow in porous media by periodic homogenization without discarding from the beginning those dynamical processes (cf. \ref{item:stationaryinterface}) on the microscale which are responsible e.g. for hysteresis phenomena (which are often excluded via an assumption like \ref{item:capillarypressure}). Decisive for the success of the approach is the unconventional scale separation assumption \eqref{eq:expansion:u}, \eqref{eq:expansion:phi} in combination with \eqref{eq:expansion:mu}, \eqref{eq:expansion:p} and the introduction of different related time scales \eqref{eq:expansion:time}, which seems to be natural to cope with instationary processes. Both of these aspects do not seem to be easily incorporated in a volume averaging approach. First numerical simulations have shown that the model (at the pore scale) exhibits phenomena not present in the standard model, e.g. hysteresis is an imbibition-drainage cycle. The full model has the form of a micro-macro model, which although demanding seems to be numerically tractable. Such an algorithmic realisation is on the agenda in the near future, in particular exploiting the massive parallelism in an iterative treatment. Further developments of the model concern the boundary conditions, i.e. removing the unnatural contact angle up to a dynamic contact angle. Furthermore the micro-macro structure shows adaptation potential to incorporate it only locally there, ``where needed''. In general we expect that our more flexible space/time multiscale expansion approach to be applicable to further and more complex problems, e.g. including transport and possibly geometry changing surface reactions.

\section*{Acknowledgement} Parts of the presented results were achieved while S. Metzger was at the Illinois Institute of Technology, Chicago and was supported by the NSF grant number NSF-DMS 1759536.

\bibliographystyle{amsplain}
\providecommand{\bysame}{\leavevmode\hbox to3em{\hrulefill}\thinspace}
\providecommand{\MR}{\relax\ifhmode\unskip\space\fi MR }
\providecommand{\MRhref}[2]{%
  \href{http://www.ams.org/mathscinet-getitem?mr=#1}{#2}
}
\providecommand{\href}[2]{#2}

\end{document}

%% file: header.tex
\makeatletter
\newcommand{\labitem}[2]{%
\def\@itemlabel{\textbf{#1}}
\item
\def\@currentlabel{#1}\label{#2}}
\makeatother

\newcommand{\abs}[1]{\left|{#1}\right|}

\newcommand{\rkla}[1]{{\left(#1\right)}}
\newcommand{\trkla}[1]{{(#1)}}
\newcommand{\gkla}[1]{{\left\{#1\right\}}}
\newcommand{\tgkla}[1]{{\{#1\}}}

\newcommand{\ekla}[1]{{\left[#1\right]}}
\newcommand{\tekla}[1]{{[#1]}}
\newcommand{\tabs}[1]{|{#1}|}

\newcommand{\jump}[1]{\left[\!\left[#1\right]\!\right]}

\newcommand{\bs}[1]{\boldsymbol{#1}}

\newcommand{\mb}[1]{\mathbf{#1}}

\newcommand{\iOmega}{\int_{\Omega}}

\newcommand{\para}[1]{\partial _{#1}}

\newcommand{\dy}{\; \mathrm{d}\mathbf{y}}

\newcommand{\nablax}{\nabla_{\mb{x}}}
\newcommand{\nablay}{\nabla_{\mb{y}}}

\renewcommand{\div}{\operatorname{div}}
\newcommand{\divx}{\operatorname{div}_{\mb{x}}}
\newcommand{\divy}{\operatorname{div}_{\mb{y}}}

\newcommand{\taufast}{\tau_{-1}}
\newcommand{\tauslow}{\tau_0}
\newcommand{\dtaufast}{\partial_{\taufast}}
\newcommand{\dtauslow}{\partial_{\tauslow}}

\newcommand{\eps}{\varepsilon}

\newcommand{\mol}[1][\empty]{
  \ifthenelse{\equal{#1}{\empty}}
  {\mathcal{J}_{\eps}}
  {\mathcal{J}_{\eps}\gkla{#1}}
}
\newcommand{\molb}[1][\empty]{
  \ifthenelse{\equal{#1}{\empty}}
  {\bs{\mathcal{J}}_{\eps}}
  {\bs{\mathcal{J}}_{\eps}\gkla{#1}}
}
\newcommand{\molbb}[1][\empty]{
  \ifthenelse{\equal{#1}{\empty}}
  {\bs{\mathfrak{J}}_{\eps}}
  {\bs{\mathfrak{J}}_{\eps}\gkla{#1}}
}

\newcommand{\molh}[1][\empty]{
  \ifthenelse{\equal{#1}{\empty}}
  {\mathcal{J}_{h}}
  {\mathcal{J}_{h}\gkla{#1}}
}
\newcommand{\molhb}[1][\empty]{
  \ifthenelse{\equal{#1}{\empty}}
  {\bs{\mathcal{J}}_{h}}
  {\bs{\mathcal{J}}_{h}\gkla{#1}}
}
\newcommand{\molhbb}[1][\empty]{
  \ifthenelse{\equal{#1}{\empty}}
  {\bs{\mathfrak{J}}_{h}}
  {\bs{\mathfrak{J}}_{h}\gkla{#1}}
}

\newcommand{\R}{\mathds{R}}

\newcommand{\xflow}{\mb{x}}

\newcommand{\uflow}{\mb{u}}

\newcommand{\yflow}{\mb{y}}

\newcommand{\restr}[2]{\ensuremath{
  \left.\kern-\nulldelimiterspace 
  #1 
  \vphantom{\big|} 
  \right|_{#2} 
  }}
  
\newcommand{\RefDomain}{Y}

\newcommand{\Domain}{\Omega}
\newcommand{\Solid}{\mathfrak{S}}
\newcommand{\Solideps}{\Solid^{\varepsilon}}
\newcommand{\Fluid}{\mathfrak{F}}
\newcommand{\Fluideps}{\Fluid^\varepsilon}
\newcommand{\dFluid}{\partial\Fluid}
\newcommand{\dFluideps}{\partial\Fluideps}

\newcommand{\iFluideps}{\int_{\Fluideps}}
\newcommand{\iFluid}{\int_{\Fluid}}
\newcommand{\iFluidw}{\int_{\Fluid_w}}
\newcommand{\iFluidn}{\int_{\Fluid_n}}
\newcommand{\iDF}{\int_{\Domain\times\Fluid}}
\newcommand{\idDF}{\int_{\partial\Domain\times\Fluid}}

\newcommand{\Sat}{S}
\newcommand{\Satw}{\Sat_{w}}
\newcommand{\Satn}{\Sat_{n}}
\newcommand{\rhow}{{\rho}_w}
\newcommand{\rhon}{{\rho}_n}
\newcommand{\etan}{{\eta}_n}
\newcommand{\etaw}{{\eta}_w}
\newcommand{\hatetan}{\hat{\eta}_n}
\newcommand{\hatetaw}{\hat{\eta}_w}
\newcommand{\charn}{\chi_n}
\newcommand{\charw}{\chi_w}
\newcommand{\charndelta}{\chi_n^\delta}
\newcommand{\charwdelta}{\chi_w^\delta}
\newcommand{\vw}{\overline{\mb{v}}_w}
\newcommand{\vn}{\overline{\mb{v}}_n}
\newcommand{\vwd}{\overline{\mb{v}}_w^\delta}
\newcommand{\vnd}{\overline{\mb{v}}_n^\delta}

%% file: twoPhase_porousMedia.bbl
\begin{thebibliography}{10}

\bibitem{AGG2012}
H.~Abels, H.~Garcke, and G.~Gr\"{u}n, \emph{Thermodynamically consistent, frame
  indifferent diffuse interface models for incompressible two-phase flows with
  different densities}, Mathematical Models and Methods in Applied Sciences
  \textbf{22} (2012), no.~3, 1150013.

\bibitem{Chen2018}
J.~Chen, S.~Sun, and X.~Wang, \emph{Homogenization of two-phase fluid flow in
  porous media via volume averaging}, Journal of Computational and Applied
  Mathematics \textbf{353} (2018), 265--282.

\bibitem{DalyRoose15}
K.R. Daly and T.~Roose, \emph{Homogenization of two fluid flow in porous
  media}, Proc. R. Soc. A \textbf{471} (2015), no.~2176.

\bibitem{Ding2007}
H.~Ding, P.~D.~M. Spelt, and C.~Shu, \emph{Diffuse interface model for
  incompressible two-phase flows with large density ratios}, Journal of
  Computational Physics \textbf{226} (2007), 2078--2095.

\bibitem{EckGarckeKnabner_engl}
C.~Eck, H.~Garcke, and P.~Knabner, \emph{Mathematical modeling}, Springer
  undergraduate mathematics series, Springer, Cham, 2017.

\bibitem{Federer1959}
H.~Federer, \emph{Curvature measures}, Trans. Amer. Math. Soc. \textbf{93}
  (1959), 418--491.

\bibitem{Gray_Miller2014}
W.~G. Gray and C.~T. Miller, \emph{{Introduction to the Thermodynamically
  Constrained Averaging Theory for Porous Medium System}}, Springer, 2014.

\bibitem{Gurtin1996}
M.~E. Gurtin, D.~Polignone, and J.~Vi{\~{n}}als, \emph{Two-phase binary fluids
  and immiscible fluids described by an order parameter}, Mathematical Models
  and Methods in Applied Sciences \textbf{6} (1996), 815.

\bibitem{Hassan_Gray1993}
S.~M. Hassanizadeh and W.~G. Gray, \emph{Thermodynamic basis of capillary
  pressure in porous media}, Water Resources Research \textbf{29} (1993),
  no.~10, 3389--3405.

\bibitem{Helmig1997}
R.~Helmig, \emph{{Multiphase Flow and Processes in the Subsurface. A
  Contribution on the Modeling of Hydrosystems}}, Springer, 1997.

\bibitem{Hohenberg1977}
P.~C. Hohenberg and B.~I. Halperin, \emph{Theory of dynamic critical
  phenomena}, Reviews of Modern Physics \textbf{49} (1977), no.~3, 435--479.

\bibitem{Hornung1997}
Ulrich Hornung (ed.), \emph{{Homogenization and porous media}}, Springer, 1997.

\bibitem{Joekar_atalii2008}
V.~Joekar-Niasar, S.~M. Hassanizadeh, and A.~Leijnse, \emph{{Insights into the
  Relationships Among Capillary Pressure, Saturation, Interfacial Area and
  Relative Permeability Using Pore-Network Modeling}}, Trans Porous Med (2008),
  201--219.

\bibitem{Lowengrub1998}
J.~Lowengrub and L.~Truskinovsky, \emph{Quasi-incompressible {C}ahn--{H}illiard
  fluids and topological transitions}, Proceeding of the Royal Society A
  \textbf{454} (1998), no.~1978, 2617--2654.

\bibitem{Niessner_Hassan2008}
J.~Niessner and S.~M. Hassanizadeh, \emph{A model for two-phase flow in porous
  media including fluid-fluid interfacial area}, Water Resour. Res. \textbf{44}
  (2008), no.~W08439, 1--10.

\bibitem{Onsager1931a}
L.~Onsager, \emph{Reciprocal relations in irreversible processes. {I}.}, Phys.
  Rev. \textbf{37} (1931), no.~4, 405--426.

\bibitem{Onsager1931}
L.~Onsager, \emph{{Reciprocal Relations in Irreversible Processes. {II}.}},
  Phys. Rev. \textbf{38} (1931), no.~12, 2265--2279.

\bibitem{Pavliotis_Stuart2008}
G.~A. Paviliotis and A.~M. Stuart, \emph{{Multiscale Methods. Averaging and
  Homogenization.}}, Springer, 2008.

\bibitem{Qian2006}
T.~Qian, X.~Wang, and P.~Sheng, \emph{A variational approach to the moving
  contact line hydrodynamics}, Journal of Fluid Mechanics \textbf{564} (2006),
  333--360.

\bibitem{Ray2019}
N.~Ray, J.~Oberlander, and P.~Frolkovic, \emph{Numerical investigation of a
  fully coupled micro-macro model for mineral dissolution and precipitation},
  Computatial Geosciences (2019).

\bibitem{Schmuck13}
M.~Schmuck, M.~Pradas, G.A. Pavliotis, and S.~Kalliadasis, \emph{{Derivation of
  effective macroscopic Stokes–Cahn–Hilliard equations for periodic
  immiscible flows in porous media}}, Nonlinearity \textbf{26} (2013), no.~12,
  3259.

\bibitem{Whitaker1986}
S.~Whitaker, \emph{{Flow in Porous Media II: The Governing Equations for
  Immiscible, Two-Phase Flow}}, Transport in Porous Media \textbf{1} (2019),
  105--125.

\end{thebibliography}
